\documentclass[a4paper,UKenglish]{lipics}
 
\usepackage{microtype}

\title{An Automata-Theoretic Approach to the Verification of Distributed Algorithms\footnote{Supported by LIA InForMel.}
}

\author[1]{C.~Aiswarya}
\author[2]{Benedikt Bollig}
\author[2]{Paul Gastin}
\affil[1]{Uppsala University\\\texttt{aiswarya.cyriac@it.uu.se}}
\affil[2]{LSV, ENS Cachan, CNRS, Inria\\\texttt{\{bollig,gastin\}@lsv.ens-cachan.fr}}
\authorrunning{C.\ Aiswarya, B.\ Bollig, and P.\ Gastin}

\Copyright{C.\ Aiswarya, Benedikt Bollig, and Paul Gastin}

\usepackage{xspace}
\usepackage[usenames,dvipsnames]{xcolor}
\usepackage{amsmath,amssymb}
\usepackage[pdflatex,recompilepics=false]{gastex}
\usepackage{enumerate}
\usepackage{enumitem}
\usepackage{extarrows}
\usepackage{upgreek}
\usepackage{booktabs}
\usepackage{pifont}
\usepackage{stmaryrd}
\usepackage{rotating}
\usepackage{mathrsfs}
\usepackage{cite}
\usepackage{graphicx,url}
\usepackage{pifont}
\usepackage{bbold}
\usepackage{fancyhdr}

\newcommand{\myqed}{\hfill\ensuremath{\lhd}}

\newcommand{\sem}[1]{\ensuremath{[\![#1]\!]}}

\newcommand{\Sem}[2]{\sem{#1}_{#2}}

\newcommand{\existspath}[1]{\ensuremath{\langle#1\rangle}}
\newcommand{\Existspath}[1]{\ensuremath{\langle#1\rangle}}
\newcommand{\forallpath}[1]{\ensuremath{[#1]}}

\newcommand{\ICPDL}{\ensuremath{\textup{LCPDL}}\xspace}

\newcommand{\loopform}[1]{\mathsf{loop}(#1)}

\newcommand{\test}[1]{\{#1\}?}
\newcommand{\pathaut}{\mathcal{A}}
\newcommand{\msgform}[2]{\mathit{msg}_{#1,#2}^{0,1}}
\newcommand{\updform}[2]{\mathit{upd}_{#1,#2}^{1,2}}
\newcommand{\nextform}[2]{\mathit{next}_{#1,#2}^{2,0}}

\newcommand{\localform}[2]{\mathit{loc}_{#1,#2}^{0,1}}

\newcommand{\passiveform}{\fwdcmd}
\renewcommand{\phi}{\varphi}
\newcommand{\formdef}{=}

\newcommand{\stay}{\epsilon}
\newcommand{\goleft}{\mathord{\leftarrow}}
\newcommand{\goright}{\mathord{\rightarrow}}
\newcommand{\goup}{\mathord{\uparrow}}
\newcommand{\godown}{\mathord{\downarrow}}

\newcommand{\A}{\mathcal{A}}

\newcommand{\col}[2]{#1[#2]}
\newcommand{\coord}[3]{#1[#2,#3]}
\newcommand{\pictofrun}[1]{\pict_{#1}}

\newcommand{\locform}{\phi}
\newcommand{\lcpdl}{\psi}

\usepackage{ifpdf}
\ifpdf
\usepackage[framemethod=TikZ]{mdframed}

\usepackage[colorinlistoftodos,prependcaption]{todonotes}
\fi

\newcounter{todocounter}

\makeatletter
\renewcommand{\paragraph}{\@startsection{paragraph}{6}{\z@}{2ex}{-0.7em}{\normalsize\bf}}
\makeatother

\makeatletter
\providecommand{\leftsquigarrow}{%
  \mathrel{\mathpalette\reflect@squig\relax}%
}
\newcommand{\reflect@squig}[2]{%
  \reflectbox{$\m@th#1\rightsquigarrow$}%
}
\makeatother

\theoremstyle{definition}
\newtheorem{myremark}[theorem]{Remark}

\gasset{frame=false,AHangle=30,AHlength=1.4,AHLength=1.6}

\begin{document}

\maketitle




\newcommand{\sendleft}[1]{\mathbf{left}{!}#1}
\newcommand{\sendright}[1]{\mathbf{right}{!}#1}
\newcommand{\recleft}[1]{\mathbf{left}{?}#1}
\newcommand{\recright}[1]{\mathbf{right}{?}#1}

\newcommand{\conf}{C}

\newcommand{\auxright}[4]{#1@ #2 \rightarrowtail #3@#4}
\newcommand{\auxleft}[4]{#3@#4 \leftarrowtail #1@#2}

\newcommand{\notauxright}[4]{#1@#2 \not\rightarrowtail #3@#4}
\newcommand{\notauxleft}[4]{#3@#4 \not\leftarrowtail #1@#2}

\newcommand{\nextcmd}{\,\text{;}\,}
\newcommand{\skipcmd}{\mathbf{skip}}
\newcommand{\updcmd}[2]{#1 := #2}
\newcommand{\fwdcmd}{\mathbf{fwd}}
\newcommand{\truecmd}{\mathbf{true}}
\newcommand{\gotocmd}[1]{\mathbf{goto}~#1}
\newcommand{\datrans}[6]{\langle#1\textup{:}~#2 \nextcmd #3 \nextcmd #4 \nextcmd #5 \nextcmd \gotocmd{#6}\rangle}

\newcommand{\confcomp}[6]{
\node[Nframe=y,Nw=22,Nh=8,Nmr=3,Nfill=n,fillcolor=blue!25,linecolor=blue](c)(#1,#2){
\begin{tabular}{c}
$#3$\\[-0.6ex]
#4~~~#5~~~#6
\end{tabular}
}}

\newcommand{\conffound}[6]{
\node[Nframe=y,Nw=22,Nh=8,Nmr=3,Nfill=n,fillcolor=ForestGreen!15,linecolor=FoestGreen](c)(#1,#2){
\begin{tabular}{c}
$#3$\\[-0.6ex]
#4~~~#5~~~#6
\end{tabular}
}}

\newcommand{\conffwd}[6]{
\node[Nframe=y,Nw=22,Nh=8,Nmr=3,Nfill=n,fillcolor=BrickRed!15,linecolor=red](c)(#1,#2){
\begin{tabular}{c}
$#3$\\[-0.6ex]
#4~~~#5~~~#6
\end{tabular}
}}

\newcommand{\activeone}[2]{
\node[Nframe=y,Nw=22,Nh=8,Nmr=3,Nfill=n](ialabel)(#1,#2){
\scalebox{0.85}{
\begin{tabular}{c}
$\sendright{\reg}$\\[-0.4ex]
$\recleft{\reg'}$
\end{tabular}
}
}}

\newcommand{\activezero}[2]{
\node[Nframe=y,Nw=22,Nh=16,Nmr=3,Nfill=n](ailabel)(#1,#2){
\scalebox{0.85}{
\begin{tabular}{c}
$\sendright{\reg'}$\\[-0.2ex]
$\recleft{\reg''}$\\[-0.4ex]
$\{\reg'',\reg\} < {\reg'}$\\[-0.6ex]
$\updcmd{\reg}{\reg'}$
\end{tabular}
}
}}

\newcommand{\passiveone}[2]{
\node[Nframe=y,Nw=22,Nh=16,Nmr=3,Nfill=n](aplabel)(#1,#2){
\scalebox{0.85}{
\begin{tabular}{c}
$\sendright{\reg'}$\\[-0.2ex]
$\recleft{\reg''}$\\[-0.3ex]
$\reg' < \reg$
\end{tabular}
}
}}

\newcommand{\passivetwo}[2]{
\node[Nframe=y,Nw=22,Nh=16,Nmr=3,Nfill=n](newaplabel)(#1,#2){
\scalebox{0.85}{
\begin{tabular}{c}
$\sendright{\reg'}$\\[-0.2ex]
$\recleft{\reg''}$\\[-0.3ex]
$\reg' < \reg''$
\end{tabular}
}
}}

\newcommand{\midpassiveone}[2]{
\node[Nframe=y,Nw=22,Nh=12,Nmr=3,Nfill=n](aplabel)(#1,#2){
\scalebox{0.85}{
\begin{tabular}{c}
$\sendright{\reg'}$\\[-0.2ex]
$\recleft{\reg''}$\\[-0.5ex]
$\reg' < \reg$
\end{tabular}
}
}}

\newcommand{\midpassivetwo}[2]{
\node[Nframe=y,Nw=22,Nh=12,Nmr=3,Nfill=n](newaplabel)(#1,#2){
\scalebox{0.85}{
\begin{tabular}{c}
$\sendright{\reg'}$\\[-0.2ex]
$\recleft{\reg''}$\\[-0.5ex]
$\reg' < \reg''$
\end{tabular}
}
}}

\newcommand{\passiveloop}[2]{
\node[Nframe=y,Nw=22,Nh=16,Nmr=3,Nfill=n](pplabel)(#1,#2){
\scalebox{0.85}{
\begin{tabular}{c}
$\fwdcmd$\\[-0.3ex]
$\recleft{\reg}$
\end{tabular}
}
}}

\newcommand{\smallpassiveloop}[2]{
\node[Nframe=y,Nw=22,Nh=8,Nmr=3,Nfill=n](pplabel)(#1,#2){
\scalebox{0.85}{
\begin{tabular}{c}
$\fwdcmd$\\[-0.4ex]
$\recleft{\reg}$
\end{tabular}
}
}}

\newcommand{\foundtrans}[2]{
\node[Nframe=y,Nw=22,Nh=16,Nmr=3,Nfill=n](aelabel)(#1,#2){
\scalebox{0.85}{
\begin{tabular}{c}
$\sendright{\reg'}$\\[-0.1ex]
$\recleft{\reg''}$\\[-0.4ex]
$\reg = \reg'$
\end{tabular}
}
}}


\newcommand{\pto}{\rightharpoonup}

\newcommand{\Ports}{\mathit{Ports}}
\newcommand{\Out}{\mathit{Out}}
\newcommand{\In}{\mathit{In}}
\newcommand{\port}{p}
\newcommand{\porta}{p}
\newcommand{\portb}{q}

\newcommand{\bound}{b}
\newcommand{\N}{\mathbb{N}}

\newcommand{\DA}{\mathcal{D}}

\newcommand{\reg}{\mathit{r}}
\newcommand{\idreg}{\mathit{id}}

\newcommand{\proc}{u}
\newcommand{\proca}{u}
\newcommand{\procb}{v}

\newcommand{\PosBool}[1]{\mathcal{B}^+(#1)}

\newcommand{\States}{{S}}
\newcommand{\Reg}{\mathit{Reg}}
\newcommand{\init}{s_0}
\newcommand{\Passive}{\mathit{Fwd}}
\newcommand{\Final}{{F}}
\newcommand{\Trans}{\Delta}

\newcommand{\Swa}{Q}
\newcommand{\transwa}{\longmapsto}
\newcommand{\xtranswa}[1]{\stackrel{#1}{\longmapsto}}
\newcommand{\iwa}{q_0}
\newcommand{\fwa}{q_f}

\newcommand{\translabel}{a}
\newcommand{\lab}{\alpha}

\newcommand{\activ}{\mathit{active}}
\newcommand{\passiv}{\mathit{passive}}
\newcommand{\electd}{\mathit{found}}
\newcommand{\found}{\mathit{found}}
\newcommand{\final}{\mathit{final}}

\newcommand{\leftp}{\mathsf{left}}
\newcommand{\rightp}{\mathsf{right}}

\newcommand{\outleft}{\mathit{out}_-}
\newcommand{\outright}{\mathit{out}_{+}}
\newcommand{\outp}[1]{\mathit{out}_{#1}}

\newcommand{\inleft}{\mathit{in}_{-}}
\newcommand{\inright}{\mathit{in}_{+}}
\newcommand{\inp}[1]{\mathit{in}_{#1}}

\newcommand{\true}{\mathit{true}}
\newcommand{\false}{\mathit{false}}

\newcommand{\msginprocess}{\texttt{?}}
\newcommand{\msgelected}{\surd}

\newcommand{\Lab}{\Sigma}
\newcommand{\Ext}{\Trans}
\newcommand{\df}{\ensuremath{\mathrel{\smash{\stackrel{\scriptscriptstyle{
    \textup{def}}}{=}}}}}

\newcommand{\set}[1]{[#1]}
\newcommand{\setz}[1]{[#1]_0}

\newcommand{\sym}{\mathsf{sym}}
\newcommand{\head}{\mathsf{head}}
\newcommand{\border}{\mathsf{border}}
\newcommand{\stm}{S_\mathsf{TM}}
\newcommand{\btm}{B_\mathsf{TM}}
\newcommand{\halt}{\textsc{halt}\xspace}
\newcommand{\tm}{\mathsf{TM}}
\newcommand{\Dtm}{\DA_\tm}
\newcommand{\bordertrue}{\mathsf{borderTrue}}
\newcommand{\headtrue}{\mathsf{headTrue}}


\begin{gpicture}[frame,name=symbrun,ignore]
\gasset{Nw=5,Nh=5,AHLength=1.8,AHlength=1.5}
\unitlength=0.9mm

\node[Nframe=n](pid1)(-8,19.8){$p_i=$}
\node[Nframe=y](pid1)(0,20){4}
\node[Nframe=y](pid2)(23,20){8}
\node[Nframe=y](pid3)(46,20){3}
\node[Nframe=y](pid4)(69,20){1}
\node[Nframe=y](pid5)(92,20){6}
\node[Nframe=y](pid6)(115,20){5}
\node[Nframe=y](pid7)(138,20){7}

\put(0,-5){
\node[Nframe=n](i)(-7,20.5){$i=$}
\node[Nframe=n](i1)(0,20){\scalebox{0.7}{1}}
\node[Nframe=n](i2)(23,20){\scalebox{0.7}{2}}
\node[Nframe=n](i3)(46,20){\scalebox{0.7}{3}}
\node[Nframe=n](i4)(69,20){\scalebox{0.7}{4}}
\node[Nframe=n](i5)(92,20){\scalebox{0.7}{5}}
\node[Nframe=n](i6)(115,20){\scalebox{0.7}{6}}
\node[Nframe=n](i7)(138,20){\scalebox{0.7}{7}}
}

\node[Nframe=n](conf)(-16,8){$\conf_0$}
\node[Nframe=n](conf)(-16,-8){$\conf_1$}
\node[Nframe=n](conf)(-16,-32){$\conf_2$}
\node[Nframe=n](conf)(-16,-48){$\conf_3$}
\node[Nframe=n](conf)(-16,-72){$\conf_4$}
\node[Nframe=n](conf)(-16,-88){$\conf_5$}
\node[Nframe=n](conf)(-16,-112){$\conf_6$}

\node[Nframe=n](conf)(-16,0){\rotatebox[origin=cc]{-90}{$\rightsquigarrow$}}
\node[Nframe=n](conf)(-16,-20){\rotatebox[origin=cc]{-90}{$\rightsquigarrow$}}
\node[Nframe=n](conf)(-16,-40){\rotatebox[origin=cc]{-90}{$\rightsquigarrow$}}
\node[Nframe=n](conf)(-16,-60){\rotatebox[origin=cc]{-90}{$\rightsquigarrow$}}
\node[Nframe=n](conf)(-16,-80){\rotatebox[origin=cc]{-90}{$\rightsquigarrow$}}
\node[Nframe=n](conf)(-16,-100){\rotatebox[origin=cc]{-90}{$\rightsquigarrow$}}

\drawedge[curvedepth=0,AHnb=0](pid1,pid2){}
\drawedge[curvedepth=0,AHnb=0](pid2,pid3){}
\drawedge[curvedepth=0,AHnb=0](pid3,pid4){}
\drawedge[curvedepth=0,AHnb=0](pid4,pid5){}
\drawedge[curvedepth=0,AHnb=0](pid5,pid6){}
\drawedge[curvedepth=0,AHnb=0](pid6,pid7){}
\drawedge[curvedepth=8,AHnb=0](pid1,pid7){}

\gasset{Nw=4.3,Nh=3.5}

\activeone{0}{0}
\activeone{23}{0}
\activeone{46}{0}
\activeone{69}{0}
\activeone{92}{0}
\activeone{115}{0}
\activeone{138}{0}

\put(0,-20){
\activezero{0}{0}
\passiveone{23}{0}
\activezero{46}{0}
\passivetwo{69}{0}
\passiveone{92}{0}
\activezero{115}{0}
\passiveone{138}{0}
}

\put(0,-40){
\activeone{0}{0}
\smallpassiveloop{23}{0}
\activeone{46}{0}
\smallpassiveloop{69}{0}
\smallpassiveloop{92}{0}
\activeone{115}{0}
\smallpassiveloop{138}{0}
}

\put(0,-60){
\passivetwo{0}{0}
\passiveloop{23}{0}
\passiveone{46}{0}
\passiveloop{69}{0}
\passiveloop{92}{0}
\activezero{115}{0}
\passiveloop{138}{0}
}

\put(0,-80){
\smallpassiveloop{0}{0}
\smallpassiveloop{23}{0}
\smallpassiveloop{46}{0}
\smallpassiveloop{69}{0}
\smallpassiveloop{92}{0}
\activeone{115}{0}
\smallpassiveloop{138}{0}
}

\put(0,-100){
\passiveloop{0}{0}
\passiveloop{23}{0}
\passiveloop{46}{0}
\passiveloop{69}{0}
\passiveloop{92}{0}
\foundtrans{115}{0}
\passiveloop{138}{0}
}

\confcomp{0}{8}{\activ_0}{4}{4}{4}
\confcomp{23}{8}{\activ_0}{8}{8}{8}
\confcomp{46}{8}{\activ_0}{3}{3}{3}
\confcomp{69}{8}{\activ_0}{1}{1}{1}
\confcomp{92}{8}{\activ_0}{6}{6}{6}
\confcomp{115}{8}{\activ_0}{5}{5}{5}
\confcomp{138}{8}{\activ_0}{7}{7}{7}

\confcomp{0}{-8}{\activ_1}{4}{7}{4}
\confcomp{23}{-8}{\activ_1}{8}{4}{8}
\confcomp{46}{-8}{\activ_1}{3}{8}{3}
\confcomp{69}{-8}{\activ_1}{1}{3}{1}
\confcomp{92}{-8}{\activ_1}{6}{1}{6}
\confcomp{115}{-8}{\activ_1}{5}{6}{5}
\confcomp{138}{-8}{\activ_1}{7}{5}{7}

\put(0,-24){
\confcomp{0}{-8}{\activ_0}{7}{7}{5}
\conffwd{23}{-8}{\passiv}{8}{4}{7}
\confcomp{46}{-8}{\activ_0}{8}{8}{4}
\conffwd{69}{-8}{\passiv}{1}{3}{8}
\conffwd{92}{-8}{\passiv}{6}{1}{3}
\confcomp{115}{-8}{\activ_0}{6}{6}{1}
\conffwd{138}{-8}{\passiv}{7}{5}{6}
}

\put(0,-40){
\confcomp{0}{-8}{\activ_1}{7}{6}{5}
\conffwd{23}{-8}{\passiv}{7}{4}{7}
\confcomp{46}{-8}{\activ_1}{8}{7}{4}
\conffwd{69}{-8}{\passiv}{8}{3}{8}
\conffwd{92}{-8}{\passiv}{8}{1}{3}
\confcomp{115}{-8}{\activ_1}{6}{8}{1}
\conffwd{138}{-8}{\passiv}{6}{5}{6}
}

\put(0,-64){
\conffwd{0}{-8}{\passiv}{7}{6}{8}
\conffwd{23}{-8}{\passiv}{6}{4}{7}
\conffwd{46}{-8}{\passiv}{8}{7}{6}
\conffwd{69}{-8}{\passiv}{7}{3}{8}
\conffwd{92}{-8}{\passiv}{7}{1}{3}
\confcomp{115}{-8}{\activ_0}{8}{8}{7}
\conffwd{138}{-8}{\passiv}{8}{5}{6}
}

\put(0,-80){
\conffwd{0}{-8}{\passiv}{8}{6}{8}
\conffwd{23}{-8}{\passiv}{8}{4}{7}
\conffwd{46}{-8}{\passiv}{8}{7}{6}
\conffwd{69}{-8}{\passiv}{8}{3}{8}
\conffwd{92}{-8}{\passiv}{8}{1}{3}
\confcomp{115}{-8}{\activ_1}{8}{8}{7}
\conffwd{138}{-8}{\passiv}{8}{5}{6}
}

\put(0,-104){
\conffwd{0}{-8}{\passiv}{8}{6}{8}
\conffwd{23}{-8}{\passiv}{8}{4}{7}
\conffwd{46}{-8}{\passiv}{8}{7}{6}
\conffwd{69}{-8}{\passiv}{8}{3}{8}
\conffwd{92}{-8}{\passiv}{8}{1}{3}
\conffound{115}{-8}{\found}{8}{8}{8}
\conffwd{138}{-8}{\passiv}{8}{5}{6}
}


\node[Nframe=n,Nw=1,Nh=1](p0)(30,15){}
\node[Nframe=n,Nw=1,Nh=1](p1)(30,2){}
\node[Nframe=n,Nw=1,Nh=1](p2)(37,-2){}
\node[Nframe=n,Nw=1,Nh=1](p3)(37,-26){}
\node[Nframe=n,Nw=1,Nh=1](p4)(37,-43){}
\node[Nframe=n,Nw=1,Nh=1](p5)(106,-43){}
\node[Nframe=n,Nw=1,Nh=1](p6)(106,-66){}
\node[Nframe=n,Nw=1,Nh=1](p7)(4,-66){}

\gasset{linewidth=0.5,linecolor=violet}

\drawedge[curvedepth=0,ELpos=50,ELdist=0.5,linecolor=violet,dash={1}0](p0,p1){}
\drawedge[curvedepth=0,ELpos=65,ELdist=0.5,linecolor=violet](p1,p2){\textcolor{violet}{\scalebox{1}{\ding{202}}}}
\drawedge[curvedepth=0,ELpos=42,ELdist=0.5,ELside=r,linecolor=violet](p2,p3){\textcolor{violet}{\scalebox{1}{\ding{203}}}}
\drawedge[curvedepth=0,ELpos=54,ELdist=0.5,ELside=r,linecolor=violet](p3,p4){\textcolor{violet}{\scalebox{1}{\ding{204}}}}
\drawedge[curvedepth=0,ELpos=57,ELdist=0.5,linecolor=violet](p4,p5){\textcolor{violet}{\scalebox{1}{\ding{205}}}}
\drawedge[curvedepth=0,ELpos=40,ELdist=0.5,ELside=r,linecolor=violet](p5,p6){\textcolor{violet}{\scalebox{1}{\ding{206}}}}
\drawedge[curvedepth=0,ELpos=44,ELdist=0.6,ELside=r,linecolor=violet](p6,p7){\textcolor{violet}{\scalebox{1}{\ding{207}}}}

\node[Nframe=n,Nw=0.5,Nh=0.5](p0)(99,15){}
\node[Nframe=n,Nw=0.5,Nh=0.5](p1)(99,2){}
\node[Nframe=n,Nw=0.5,Nh=0.5](p2)(106,-2){}
\node[Nframe=n,Nw=0.5,Nh=0.5](p3)(106,-25){}
\node[Nframe=n,Nw=0.5,Nh=0.5](p4)(106,-37){}
\node[Nframe=n,Nw=0.5,Nh=0.5](p5)(-9,-37){}
\node[Nframe=n,Nw=0.5,Nh=0.5](p6)(-9,-65){}

\gasset{linewidth=0.5,linecolor=blue}

\drawedge[curvedepth=0,ELpos=50,ELdist=0.5,linecolor=blue,dash={1}0](p0,p1){}
\drawedge[curvedepth=0,ELpos=65,ELdist=0.5](p1,p2){\textcolor{blue}{\scalebox{1}{\ding{202}}}}
\drawedge[curvedepth=0,ELpos=43,ELside=r,ELdist=0.5](p2,p3){\textcolor{blue}{\scalebox{1}{\ding{203}}}}
\drawedge[curvedepth=0,ELpos=23,ELside=r,ELdist=0.5](p3,p4){\textcolor{blue}{\scalebox{1}{\ding{204}}}}
\drawedge[curvedepth=0,ELpos=42,ELside=r,ELdist=0.45](p4,p5){\textcolor{blue}{\scalebox{1}{\ding{205}}}}
\drawedge[curvedepth=0,ELpos=52,ELside=r,ELdist=0.5](p5,p6){\textcolor{blue}{\scalebox{1}{\ding{206}}}}


\unitlength=1.2mm
\put(-2,-103){
\node[Nframe=n,Nw=1,Nh=0](p1)(-3,1){\textcolor{violet}{$\pathaut_{r''}^{1}$:}}
\node[Nframe=n,Nw=1,Nh=0](p1)(2,0){}
\node[Nframe=n,Nw=1,Nh=0](p2)(10,0){}
\node[Nframe=n,Nw=1,Nh=0](p3)(20,0){}
\node[Nframe=n,Nw=1,Nh=0](p4)(30,0){}
\node[Nframe=n,Nw=1,Nh=0](p5)(40,0){}
\node[Nframe=n,Nw=1,Nh=0](p6)(50,0){}
\node[Nframe=n,Nw=1,Nh=0](p7)(60,0){}
\node[Nframe=n,Nw=1,Nh=0](p8)(70,0){}
\node[Nframe=n,Nw=1,Nh=0](p9)(80,0){}
\node[Nframe=n,Nw=1,Nh=0](p10)(90,0){}
\node[Nframe=n,Nw=1,Nh=0](p11)(100,0){}
\node[Nframe=n,Nw=1,Nh=0](p12)(110,0){}

\gasset{linewidth=0.3,linecolor=violet,AHangle=25}

\drawedge[curvedepth=0,ELside=r,ELpos=50,ELdist=0.5,linewidth=0,linecolor=violet,dash={1}0](p1,p2){}
\drawedge[curvedepth=0,ELside=r,ELpos=50,ELdist=0.5,linewidth=0,linecolor=violet](p2,p3){\textcolor{violet}{\scalebox{1}{\ding{202}}}}
\drawedge[curvedepth=0,ELside=r,ELpos=50,ELdist=0.5,linewidth=0,linecolor=violet](p4,p5){\textcolor{violet}{\scalebox{1}{\ding{203}}}}
\drawedge[curvedepth=0,ELside=r,ELpos=50,ELdist=0.5,linewidth=0,linecolor=violet](p7,p8){\textcolor{violet}{\scalebox{1}{\ding{204}}}}
\drawedge[curvedepth=0,ELside=r,ELpos=50,ELdist=0.5,linewidth=0,linecolor=violet](p8,p9){\textcolor{violet}{\scalebox{1}{\ding{205}}}}
\drawedge[curvedepth=0,ELside=r,ELpos=50,ELdist=0.5,linewidth=0,linecolor=violet](p10,p11){\textcolor{violet}{\scalebox{1}{\ding{206}}}}
\drawedge[curvedepth=0,ELside=r,ELpos=50,ELdist=0.5,linewidth=0,linecolor=violet](p11,p12){\textcolor{violet}{\scalebox{1}{\ding{207}}}}

\drawedge[curvedepth=0,ELpos=50,linecolor=violet,dash={1}0](p1,p2){}
\drawedge[curvedepth=0,ELpos=50,linecolor=violet](p2,p3){\textcolor{violet}{\scalebox{0.9}{$\msgform{r}{r'}$}}}
\drawedge[curvedepth=0,ELpos=50,linecolor=violet](p3,p4){\textcolor{violet}{\scalebox{0.9}{$\updform{r'}{r'}$}}}
\drawedge[curvedepth=0,ELpos=50,linecolor=violet](p4,p5){\textcolor{violet}{\scalebox{0.9}{$\nextform{r'}{r'}$}}}
\drawedge[curvedepth=0,ELpos=50,linecolor=violet](p5,p6){\textcolor{violet}{\scalebox{0.9}{$\localform{r'}{r'}$}}}
\drawedge[curvedepth=0,ELpos=50,linecolor=violet](p6,p7){\textcolor{violet}{\scalebox{0.9}{$\updform{r'}{r}$}}}
\drawedge[curvedepth=0,ELpos=50,linecolor=violet](p7,p8){\textcolor{violet}{\scalebox{0.9}{$\nextform{r}{r}$}}}
\drawedge[curvedepth=0,ELpos=50,linecolor=violet](p8,p9){\textcolor{violet}{\scalebox{0.9}{$\msgform{r}{r'}$}}}
\drawedge[curvedepth=0,ELpos=50,linecolor=violet](p9,p10){\textcolor{violet}{\scalebox{0.9}{$\updform{r'}{r'}$}}}
\drawedge[curvedepth=0,ELpos=50,linecolor=violet](p10,p11){\textcolor{violet}{\scalebox{0.9}{$\nextform{r'}{r'}$}}}
\drawedge[curvedepth=0,ELpos=50,linecolor=violet](p11,p12){\textcolor{violet}{\scalebox{0.9}{$\msgform{r'}{r''}$}}}
}


\unitlength=1.2mm
\put(-2,-95){
\node[Nframe=n,Nw=1,Nh=0](p1)(-3,1){\textcolor{blue}{$\pathaut_{r'}^{1}$:}}
\node[Nframe=n,Nw=1,Nh=0](p1)(2,0){}
\node[Nframe=n,Nw=1,Nh=0](p2)(10,0){}
\node[Nframe=n,Nw=1,Nh=0](p3)(20,0){}
\node[Nframe=n,Nw=1,Nh=0](p4)(30,0){}
\node[Nframe=n,Nw=1,Nh=0](p5)(40,0){}
\node[Nframe=n,Nw=1,Nh=0](p6)(50,0){}
\node[Nframe=n,Nw=1,Nh=0](p7)(60,0){}
\node[Nframe=n,Nw=1,Nh=0](p8)(70,0){}
\node[Nframe=n,Nw=1,Nh=0](p9)(80,0){}
\node[Nframe=n,Nw=1,Nh=0](p10)(90,0){}
\node[Nframe=n,Nw=1,Nh=0](p11)(100,0){}
\node[Nframe=n,Nw=1,Nh=0](p12)(110,0){}

\gasset{linewidth=0.3,linecolor=blue,AHangle=25}

\drawedge[curvedepth=0,ELside=r,ELpos=50,ELdist=0.5,linewidth=0,linecolor=violet,dash={1}0](p1,p2){}
\drawedge[curvedepth=0,ELside=r,ELpos=50,ELdist=0.5,linewidth=0,linecolor=violet](p2,p3){\textcolor{blue}{\scalebox{1}{\ding{202}}}}
\drawedge[curvedepth=0,ELside=r,ELpos=50,ELdist=0.5,linewidth=0,linecolor=violet](p4,p5){\textcolor{blue}{\scalebox{1}{\ding{203}}}}
\drawedge[curvedepth=0,ELside=r,ELpos=50,ELdist=0.5,linewidth=0,linecolor=violet](p7,p8){\textcolor{blue}{\scalebox{1}{\ding{204}}}}
\drawedge[curvedepth=0,ELside=r,ELpos=50,ELdist=0.5,linewidth=0,linecolor=violet](p8,p9){\textcolor{blue}{\scalebox{1}{\ding{205}}}}
\drawedge[curvedepth=0,ELside=r,ELpos=50,ELdist=0.5,linewidth=0,linecolor=violet](p10,p11){\textcolor{blue}{\scalebox{1}{\ding{206}}}}

\drawedge[curvedepth=0,ELpos=50,linecolor=blue,dash={1}0](p1,p2){}
\drawedge[curvedepth=0,ELpos=50,linecolor=blue](p2,p3){\textcolor{blue}{\scalebox{0.9}{$\msgform{r}{r'}$}}}
\drawedge[curvedepth=0,ELpos=50,linecolor=blue](p3,p4){\textcolor{blue}{\scalebox{0.9}{$\updform{r'}{r'}$}}}
\drawedge[curvedepth=0,ELpos=50,linecolor=blue](p4,p5){\textcolor{blue}{\scalebox{0.9}{$\nextform{r'}{r'}$}}}
\drawedge[curvedepth=0,ELpos=50,linecolor=blue](p5,p6){\textcolor{blue}{\scalebox{0.9}{$\localform{r'}{r'}$}}}
\drawedge[curvedepth=0,ELpos=50,linecolor=blue](p6,p7){\textcolor{blue}{\scalebox{0.9}{$\updform{r'}{r}$}}}
\drawedge[curvedepth=0,ELpos=50,linecolor=blue](p7,p8){\textcolor{blue}{\scalebox{0.9}{$\nextform{r}{r}$}}}
\drawedge[curvedepth=0,ELpos=50,linecolor=blue](p8,p9){\textcolor{blue}{\scalebox{0.9}{$\msgform{r}{r'}$}}}
\drawedge[curvedepth=0,ELpos=50,linecolor=blue](p9,p10){\textcolor{blue}{\scalebox{0.9}{$\updform{r'}{r'}$}}}
\drawedge[curvedepth=0,ELpos=50,linecolor=blue](p10,p11){\textcolor{blue}{\scalebox{0.9}{$\nextform{r'}{r'}$}}}
\drawedge[curvedepth=0,ELpos=50,linecolor=blue](p11,p12){\textcolor{blue}{\scalebox{0.9}{$\localform{r'}{r'}$}}}
}

\end{gpicture}


\vspace{-4ex}
\begin{abstract}
{\textbf{Abstract.}} We introduce an automata-theoretic method for the verification of distributed algorithms running on ring networks. In a distributed algorithm, an arbitrary number of processes cooperate to achieve a common goal (e.g., elect a leader). Processes have unique identifiers (pids) from an infinite, totally ordered domain. An algorithm proceeds in synchronous rounds, each round allowing a process to perform a bounded sequence of actions such as send or receive a pid, store it in some register, and compare register contents wrt.\ the associated total order. An algorithm is supposed to be correct independently of the number of processes. To specify correctness properties, we introduce a logic that can reason about processes and pids. Referring to leader election, it may say that, at the end of an execution, each process stores the maximum pid in some dedicated register. Since the verification of distributed algorithms is undecidable, we propose an underapproximation technique, which bounds the number of rounds. This is an appealing approach, as the number of rounds needed by a distributed algorithm to conclude is often exponentially smaller than the number of processes. We provide an automata-theoretic solution, reducing model checking to emptiness for alternating two-way automata on words. Overall, we show that round-bounded verification of distributed algorithms over rings is PSPACE-complete.
\end{abstract}

\pagestyle{fancy}
\fancyhead{}
\renewcommand{\headrulewidth}{0pt}
\fancyfoot[C]{\vspace{2ex}\thepage}


\section{Introduction}

Distributed algorithms are a classic discipline of computer science and continue to be an active field of research \cite{Lynch:1996,Fokkink2013}. A distributed algorithm employs several processes, which perform one and the same program to achieve a common goal. It is required to be correct independently of the number of processes. Prominent examples are leader-election algorithms, whose task is to determine a unique leader process and to announce it to all other processes. Those algorithms are often studied for ring architectures. One practical motivation comes from local-area networks that are based on a token-ring protocol. Moreover, rings generally allow one to nicely illustrate the main conceptual ideas of an algorithm.

However, it is well-known that there is no (deterministic) distributed algorithm over rings that elects a leader under the assumption of anonymous processes. Therefore, classical algorithms, such as Franklin's algorithm \cite{Franklin:1982} or the Dolev-Klawe-Rodeh algorithm \cite{DolevKR82}, assume that every process is equipped with a unique process identifier (pid) from an infinite, totally ordered domain.
In this paper, we consider such distributed algorithms, which work on ring architectures and can access unique pids as well as the associated total order.

Distributed algorithms are intrinsically hard to analyze. Correctness proofs are often intricate and use subtle inductive arguments. Therefore, it is worthwhile to consider automatic verification methods such as model checking \cite{ClarkeGP2001}. Besides a formal model of an algorithm, this requires a generic specification language that is feasible from an algorithmic point of view but expressive enough to formulate correctness properties. In this paper, we propose a language that can reason about processes, states, and pids. In particular, it will allow us to formalize when a leader-election algorithm is correct: \emph{At the end of an execution, every process stores, in register $r$, the maximum pid among all processes}.
Our language is inspired by Data-XPath, which can reason about trees over infinite alphabets \cite{BenediktFG08,BojanczykMSS09,FS11}.

However, formal verification of distributed algorithms cumulates various difficulties that already arise, separately, in more standard verification: First, the number of processes is unknown, which amounts to parameterized verification \cite{Esparza14}; second, processes manipulate data from an infinite domain \cite{BojanczykMSS09,FS11}. In each case, even simple verification questions are undecidable, and so is the combination of both.

In various other contexts, a successful approach to retrieving decidability has been a form of \emph{bounded model checking}. The idea is to consider correctness up to some parameter, which restricts the set of runs of the algorithm in a non-trivial way. In multi-threaded recursive programs, for example, one may restrict the number of control switches between different threads \cite{Qadeer:TACAS05}. Actually, this idea seems even more natural in the context of distributed algorithms, which usually proceed in \emph{rounds}. In each round, a process may emit some messages (here: pids) to its neighbors, and then receive messages from its neighbors. Pids can be stored in registers, and a process can check the relation between stored pids before it moves to a new state and is ready for a new round. It turns out that the number of rounds is often exponentially smaller than the number of processes (cf.\ the above-mentioned leader-election algorithms). Thus, roughly speaking, a small number of rounds allows us to verify correctness of an algorithm for a large number of processes.

The key idea of our method is to interpret a (round-bounded) execution of a distributed algorithm symbolically as a word-like structure over a finite alphabet. The finite alphabet is constituted by the transitions that occur in the algorithm and possibly contain tests of pids wrt.\ equality or the associated total order. To determine feasibility of a symbolic execution (i.e., \emph{is there a ring that satisfies all the guards employed?}), we use propositional dynamic logic with loop and converse (LCPDL) over words \cite{Goeller2009}. Basically, we translate a given distributed algorithm into a formula that detects cyclic (i.e., contradictory) smaller-than tests. Its models are precisely the feasible symbolic executions. A specification is translated into LCPDL as well so that verification amounts to checking satisfiability of a single formula. The latter can be reduced to an emptiness problem for alternating two-way automata over words so that we obtain a PSPACE procedure for round-bounded model checking.

\paragraph{Related Work.}

Considerable effort has been devoted to the verification of fault-tolerant algorithms, which have to cope with faults such as lost or corrupted messages (e.g., \cite{Merz:2009,KonnovVW14}). After all, there have been only very few generic approaches to model checking distributed algorithms. In \cite{KVW12}, several possible reasons for this are identified, among them the presence of unbounded data types and an unbounded number of processes, which we have to treat simultaneously in our framework.  Parameterized model checking of ring-based systems where communication is subject to a token policy and the message alphabet is finite has been studied in \cite{EmersonN03,AminofJKR14}.

The theory of words and trees over infinite alphabets (aka data words/trees) provides an elegant formal framework for database-related notions such as XML documents \cite{BojanczykMSS09}, or for the analysis of programs with data structures such as lists and arrays \cite{Alur:2011,Alur:2012}. Notably, streaming transducers \cite{Alur:2011} also work over an infinite, totally ordered domain. The difference to our work is that we model distributed algorithms and provide a logical specification language. Recall that the latter borrows concepts from \cite{BenediktFG08,BojanczykMSS09,FS11}, whose logic is designed to reason about XML documents. A fragment of MSO logic over \emph{ordered} data trees was studied in \cite{Tan14}. The paper \cite{BCGK-fossacs12} pursued a symbolic model-checking approach to systems involving data. But the model was purely sequential and pids could only be compared for equality. The ordering on the data domain actually has a subtle impact on the choice of the specification language.

\paragraph{Outline.}

In Section,~\ref{sec:algorithms}, we present our model of a distributed algorithm.
Section~\ref{sec:spec} introduces the specification language to express correctness criteria. In Section~\ref{sec:verification}, we show how to solve the round-bounded model-checking problem  in polynomial space. We conclude in Section~\ref{sec:conclusion}.
Some proof details are omitted but can be found in the appendix.



\section{Distributed Algorithms}\label{sec:algorithms}

\newcommand{\lef}{\mathbf{left}}
\newcommand{\rig}{\mathbf{right}}
\newcommand{\idsend}{\mathit{snd}}
\newcommand{\idrec}{\mathit{rec}}
\newcommand{\update}{\mathit{upd}}
\newcommand{\idguard}{\mathit{grd}}
\newcommand{\ring}{\mathcal{R}}
\newcommand{\Pos}[2]{\mathit{Pos}(#1,#2)}
\newcommand{\Coord}[1]{\mathit{Pos}(#1)}

By $\N = \{0,1,2,\ldots\}$, we denote the set of natural numbers. For $n \in \N$, we set $\set{n} = \{1,\ldots,n\}$ and $\setz{n} = \{0,1,\ldots,n\}$.
The set of finite words over an alphabet $A$ is denoted by $A^\ast$, and the set of nonempty finite words by $A^+$.

\paragraph{Syntax of Distributed Algorithms.}

We consider distributed algorithms that run on arbitrary ring architectures. A ring consists of a finite number of processes, each having a unique process identifier (pid). Every process has a unique left neighbor (referred to by $\lef$) and a unique right neighbor (referred to by $\rig$). Formally, a \emph{ring} is a tuple $\ring=(n:p_1,\ldots,p_n)$, given by its size $n \ge 1$ and the pids $p_i \in \N$ assigned to process $i \in \set{n}$. We require that pids are unique, i.e., $p_i \neq p_j$ whenever $i \neq j$. For a process $i < n$, process $i+1$ is the right neighbor of $i$. Moreover, $1$ is the right neighbor of $n$. Analogously, if $i > 2$, then $i-1$ is the left neighbor of $i$. Moreover, $n$ is the left neighbor of $1$. Thus, processes $1$ and $n$ must not be considered as the ``first'' or ``last'' process. Actually, a distributed algorithm will not be able to distinguish between, for example, $(4:4,1,5,2)$ and $(4:5,2,4,1)$.

One given distributed algorithm can be run on \emph{any} ring. It is given by a single program $\DA$, and each process runs a copy of $\DA$. It is convenient to think of $\DA$ as a (finite) automaton. Processes proceed in synchronous rounds. In one round, every process executes one transition of its program. In addition to the change of state, a process may optionally perform the following phases within a transition: (i) send some pids to its neighbors, (ii) receive pids from its neighbors and store them in registers, (iii) compare register contents with one another, (iv) update its registers. For example, consider the transition $t=\datrans{s}{\sendleft{r} \nextcmd \sendright{r'}}{\recright{r'}}{r<r'}{\updcmd{r}{r'}}{s'}$. A process can execute $t$ if it is in state $s$. It then sends the contents of register $r$ to its left neighbor and the contents of $r'$ to its right neighbor. If, afterwards, it receives a pid $p$ from its right neighbor, it stores $p$ in $r'$. If $p$ is greater than what has been stored in $r$, it sets $r$ to $p$ and goes to state $s'$. Otherwise, the transition is not applicable.
The first phase can, alternatively, be filled with a special command $\fwdcmd$. Then, a process will just forward any pid it receives. Note that a message can be forwarded, in one and the same round, across several processes executing $\fwdcmd$.

\begin{definition}\label{def:da}
A \emph{distributed algorithm} $\DA=(\States,\init,\Reg,\Trans)$ consists of
a nonempty finite set $\States$ of \emph{(local) states},
an \emph{initial state} $\init \in \States$,
a nonempty finite set $\Reg$ of \emph{registers}, and
a nonempty finite set $\Trans$ of \emph{transitions}. 
A transition is of the form
$\datrans{s}{\mathit{send}}{\mathit{rec}}{\mathit{guard}}{\mathit{update}}{s'}$
where $s,s' \in \States$ and the components $\mathit{send}$, $\mathit{rec}$, $\mathit{guard}$, and $\mathit{update}$ are built as follows: 
\begin{itemize}\itemsep=0.4ex
\item[] $\mathit{send} ~::=~ \skipcmd ~\mid~ \fwdcmd ~\mid~ \sendleft{r} ~\mid~ \sendright{r} ~\mid~ \sendleft{r}\nextcmd \sendright{r'}$

\item[] $\mathit{rec} ~::=~ \skipcmd ~\mid~ \recleft{r} ~\mid~ \recright{r} ~\mid~ \recleft{r}\nextcmd \recright{r'}$

\item[] $\mathit{guard} ~::=~ \skipcmd ~\mid~ r < r' ~\mid~ r = r' ~\mid~ \mathit{guard} \nextcmd \mathit{guard}$
\item[] $\mathit{update} ~::=~ \skipcmd ~\mid~ \updcmd{r}{r'} ~\mid~ \mathit{update}\nextcmd\mathit{update}$
\end{itemize}
with $r$ and $r'$ ranging over $\Reg$. We require that
\begin{itemize}
\item[(1)] in a $\mathit{rec}$ statement of the form $\recleft{r}\nextcmd \recright{r'}$, we have $r \neq r'$ (actually, the order of the two receive actions does not matter), and
\item[(2)] in an $\mathit{update}$ statement, every register occurs at most once as a left-hand side.
\end{itemize}
In the following, occurrences of ``$\skipcmd\,\text{;}$'' are omitted; this does not affect the semantics.
\myqed
\end{definition}

Note that a guard $r \le r'$ can be simulated in terms of guards $r < r'$ and $r = r'$, using several transitions. We separate $<$ and $=$ for convenience. They are actually quite different in nature, as we will see later in the proof of our main result.

At the beginning of an execution of an algorithm, every register contains the pid of the respective process. We also assume, wlog., that there is a special register $\idreg \in \Reg$ that is never updated, i.e., no transition contains a command of the form $\recleft{\idreg}$, $\recright{\idreg}$, or $\updcmd{\idreg}{r}$. A process can thus, at any time, access its own pid in terms of $\idreg$.

In the semantics, we will suppose that all updates of a transition happen simultaneously, i.e., after executing $\updcmd{r}{r'} \nextcmd \updcmd{r'}{r}$, the values previously stored in $r$ and $r'$ will be swapped (and do not necessarily coincide).
As, moreover, the order of two sends and the order of two receives within a 
transition do not matter, this will allow us to identify a transition with 
the set of states, commands (apart from $\skipcmd$), and guards that it contains. For example, $t=\datrans{s}{\sendleft{r} \nextcmd \sendright{r'}}{\recright{r'}}{r<r'}{\updcmd{r}{r'}}{s'}$ is considered as the set $t=\{s\,,\,\sendleft{r}\,,\,\sendright{r'}\,,\,\recright{r'}\,,\,r<r'\,,\,\updcmd{r}{r'}\,,\,\gotocmd{s'}\}$.


\begin{figure}[t]
\centering
\parbox{\textwidth}{
\scalebox{0.85}{
$\begin{array}{lcl}
\textbf{states: } \activ,\passiv & & t_1 = \langle \activ\textup{:}~\sendleft{\idreg} \nextcmd \sendright{\idreg} \nextcmd \recleft{r_1} \nextcmd \recright{r_2} \nextcmd r_1 < \idreg \nextcmd r_2 < \idreg \nextcmd \gotocmd{\activ}\rangle\\[0.5ex]

\phantom{\textbf{states: }} {\found} & & t_2 = \langle \activ\textup{:}~\rule{\widthof{$\sendleft{\idreg} \nextcmd \sendright{\idreg} \nextcmd \recleft{r_1} \nextcmd \recright{r_2}$}}{0.4pt} \nextcmd \idreg < r_1 \nextcmd \gotocmd{\passiv}\rangle\\[0.5ex]

\textbf{initial state: } \activ & & t_3 = \langle \activ\textup{:}~\rule{\widthof{$\sendleft{\idreg} \nextcmd \sendright{\idreg} \nextcmd \recleft{r_1} \nextcmd \recright{r_2}$}}{0.4pt} \nextcmd \idreg < r_2 \nextcmd \gotocmd{\passiv}\rangle\\[0.5ex]

\textbf{registers: } \idreg,r,r_1,r_2 & & t_4 = \langle \activ\textup{:}~\rule{\widthof{$\sendleft{\idreg} \nextcmd \sendright{\idreg} \nextcmd \recleft{r_1} \nextcmd \recright{r_2}$}}{0.4pt} \nextcmd \idreg = r_1 \nextcmd \updcmd{r}{\idreg} \nextcmd \gotocmd{\found}\rangle\\[0.5ex]

& & t_5 = \langle \passiv\textup{:}~\fwdcmd \nextcmd \recleft{r}\nextcmd \gotocmd{\passiv}\rangle
\end{array}$
}}
\caption{Franklin's leader-election algorithm $\DA_\mathsf{Franklin}$\label{fig:franklin}}
\hspace{3ex}
\centering
\parbox{\textwidth}{
\scalebox{0.85}{
$\begin{array}{lcl}
\textbf{states: } \activ_0,\activ_1 & & t_1 = \langle \activ_0\textup{:}~\sendright{r} \nextcmd \recleft{r'} \nextcmd \gotocmd{\activ_1}\rangle\\[0.5ex]

\phantom{\textbf{states: }} {\passiv,\found} & & t_2 = \langle \activ_1\textup{:}~\sendright{r'} \nextcmd \recleft{r''} \nextcmd r'' < r' \nextcmd r < r' \nextcmd \updcmd{r}{r'} \nextcmd \gotocmd{\activ_0}\rangle\\[0.5ex]

\textbf{initial state: } \activ_0 & & t_3 = \langle \activ_1\textup{:}~\rule{\widthof{$\sendright{r'} \nextcmd \recleft{r''}$}}{0.4pt} \nextcmd r' < r \nextcmd \gotocmd{\passiv}\rangle\\[0.5ex]

\textbf{registers: } \idreg,r,r',r'' & & t_4 = \langle \activ_1\textup{:}~\rule{\widthof{$\sendright{r'} \nextcmd \recleft{r''}$}}{0.4pt} \nextcmd r' < r'' \nextcmd \gotocmd{\passiv}\rangle\\[0.5ex]

& & t_5 = \langle \activ_1\textup{:}~\rule{\widthof{$\sendright{r'} \nextcmd \recleft{r''}$}}{0.4pt} \nextcmd r = r' \nextcmd \gotocmd{\found}\rangle\\[0.5ex]

\phantom{\textbf{states: } \activ,\passiv,\found} & & t_6 = \langle \passiv\textup{:}~\fwdcmd \nextcmd \recleft{r}\nextcmd \gotocmd{\passiv}\rangle
\end{array}$
}}
\caption{Dolev-Klawe-Rodeh leader-election algorithm $\DA_\mathsf{DKR}$\label{fig:dkr}}
\end{figure}

Before defining the semantics of a distributed algorithm, we will look at two examples.

\begin{example}[Franklin's Leader-Election Algorithm]\label{ex:franklin}
Consider Franklin's algorithm $\DA_\mathsf{Franklin}$ to determine a leader in a ring \cite{Franklin:1982}. It is given in Figure~\ref{fig:franklin}.
The goal is to assign leadership to the process with the highest pid. To do so, every process sends its own pid to both neighbors, receives the pids of its left and right neighbor, and stores them in registers $r_1$ and $r_2$, respectively (transitions $t_1,\ldots,t_4$). If a process is a local maximum, i.e., $r_1 < \idreg$ and $r_2 < \idreg$ hold, it is still in the race for leadership and stays in state $\activ$. Otherwise, it has to take $t_2$ or $t_3$ and goes into state $\passiv$. 
In $\passiv$, a process will just forward any pid it receives and store the 
message coming from the left in $r$ (transition $t_5$).
When an active process receives its own pid (transition $t_4$), it knows it is the only remaining active process. It copies its own pid into $r$, which henceforth refers to the leader. 
We may say that a run is accepting (or terminating) when all processes terminate in $\passiv$ or $\found$. Then, at the end of any accepting run, (i) there is exactly one process $i_0$ that terminates in $\found$, (ii) all processes store the pid of $i_0$ in register $r$, and the pid of $i_0$ is the maximum of all pids in the ring. Since, in every round, at least half of the active processes become passive, the algorithm terminates after at most $\lfloor \log_2 n\rfloor +1$ rounds where $n$ is the number of processes.
\myqed
\end{example}

\begin{example}[Dolev-Klawe-Rodeh Leader-Election Algorithm]\label{ex:dolev}
The Dolev-Klawe-Rodeh leader-election algorithm \cite{DolevKR82} is an
adaptation of Franklin's algorithm to cope with unidirectional rings, where a
process can only, say, send to the right and receive from the left.  The
algorithm, denoted $\DA_\mathsf{DKR}$, is given in Figure~\ref{fig:dkr}.
The idea is that the local maximum among the processes $i-2,i-1,i$ is determined
by $i$ (rather than $i-1$). Therefore, each process $i$ will execute two
transitions, namely $t_1$ and $t_2$, and store the pids sent by $i-2$ and $i-1$ in
$r''$ and $r'$, respectively. After two rounds, since $r$ still contains the pid of
$i$ itself, $i$ can test if $i-1$ is a local maximum among $i-2,i-1,i$ using the guards in
transition $t_2$.  If both guards are satisfied, $i$ stores the pid sent by $i-1$ 
in $r$. It henceforth ''represents'' process $i-1$, which is still in the race, and goes to state $\activ_0$.
Otherwise, it enters $\passiv$, which has the same task as in Franklin's algorithm.  The algorithm is correct in the following sense: At the end of an accepting run (each process ends in $\passiv$ or
$\found$), (i) there is exactly one process that terminates in $\found$ (but not
necessarily the one with the highest pid), and (ii) all processes store the
maximal pid in register $r$.  The algorithm terminates after at most $2\lfloor
\log_2 n\rfloor+2$ rounds.  Note that the correctness of $\DA_\mathsf{DKR}$ is
less clear than that of $\DA_\mathsf{Franklin}$.  \myqed
\end{example}

\newcommand{\height}{\parbox[0pt][3.6ex][c]{0cm}{}}
\makeatletter
\def\hlinewd#1{%
  \noalign{\ifnum0=`}\fi\hrule \@height #1 \futurelet
   \reserved@a\@xhline}
\makeatother


\newcommand{\regmap}{\rho}
\newcommand{\run}{\chi}
\newcommand{\confrel}[1]{\stackrel{#1}{\rightsquigarrow}}

\paragraph{Semantics of Distributed Algorithms.}

Now, we give the formal semantics of a distributed algorithm $\DA=(\States,\init,\Reg,\Trans)$. Recall that $\DA$ can be run on any ring $\ring=(n:p_1,\ldots,p_n)$. An ($\ring$-)configuration of $\DA$ is a tuple $(s_1,\ldots,s_n,\regmap_1,\ldots,\regmap_n)$ where $s_i$ is the current state of process $i$ and $\regmap_i: \Reg \to \{p_1,\ldots,p_n\}$ maps each register to a pid. The configuration is called \emph{initial} if, for all processes $i \in \set{n}$, we have $s_i = \init$ and $\regmap_i(r) = p_i$ for all $r \in \Reg$.
Note that there is a unique initial $\ring$-configuration.

In one round, the algorithm moves from one configuration to another
one.  This is described by a relation $\conf \confrel{t}
\conf'$ where $\conf=(s_1,\ldots,s_n,\regmap_1,\ldots,\regmap_n)$ and
$\conf'=(s_1',\ldots,s_n',\regmap_1',\ldots,\regmap_n')$ are
$\ring$-configurations and $t = (t_1,\ldots,t_n) \in \Trans^n$ is a tuple of
transitions where $t_i$ is executed by process $i$. To determine when $\conf
\confrel{t} \conf'$ holds, we first define two auxiliary relations.  For
registers $r,r' \in \Reg$ and processes $i,j \in \set{n}$, we write
$\auxright{r}{i}{r'}{j}$ if the contents of $r$ is sent to the right from $i$ to $j$, where
it is stored in $r'$.  Thus, we require that
\begin{center}
$\sendright{r} \in t_i ~\wedge~ \recleft{r'} \in t_j ~\wedge~ \fwdcmd \in t_k$ for all $k \in \mathit{Between}(i,j)$
\end{center}
where $\mathit{Between}(i,j)$ means $\{i+1,\ldots,j-1\}$ if $i<j$ or 
$\{1,\ldots,j-1,i+1,\ldots,n\}$ if $j\leq i$.
Note that, due to the $\fwdcmd$ command, $\auxright{r}{i}{r'}{j}$ may hold for several $r'$ and $j$.
The meaning of $\auxleft{r}{i}{r'}{j}$ is analogous, we just replace ``right direction'' by ``left direction'':
\begin{center}
$\sendleft{r} \in t_i ~\wedge~ \recright{r'} \in t_j ~\wedge~ \fwdcmd \in t_k$ for all $k \in \mathit{Between}(j,i)$.
\end{center}

\begin{figure}[t]
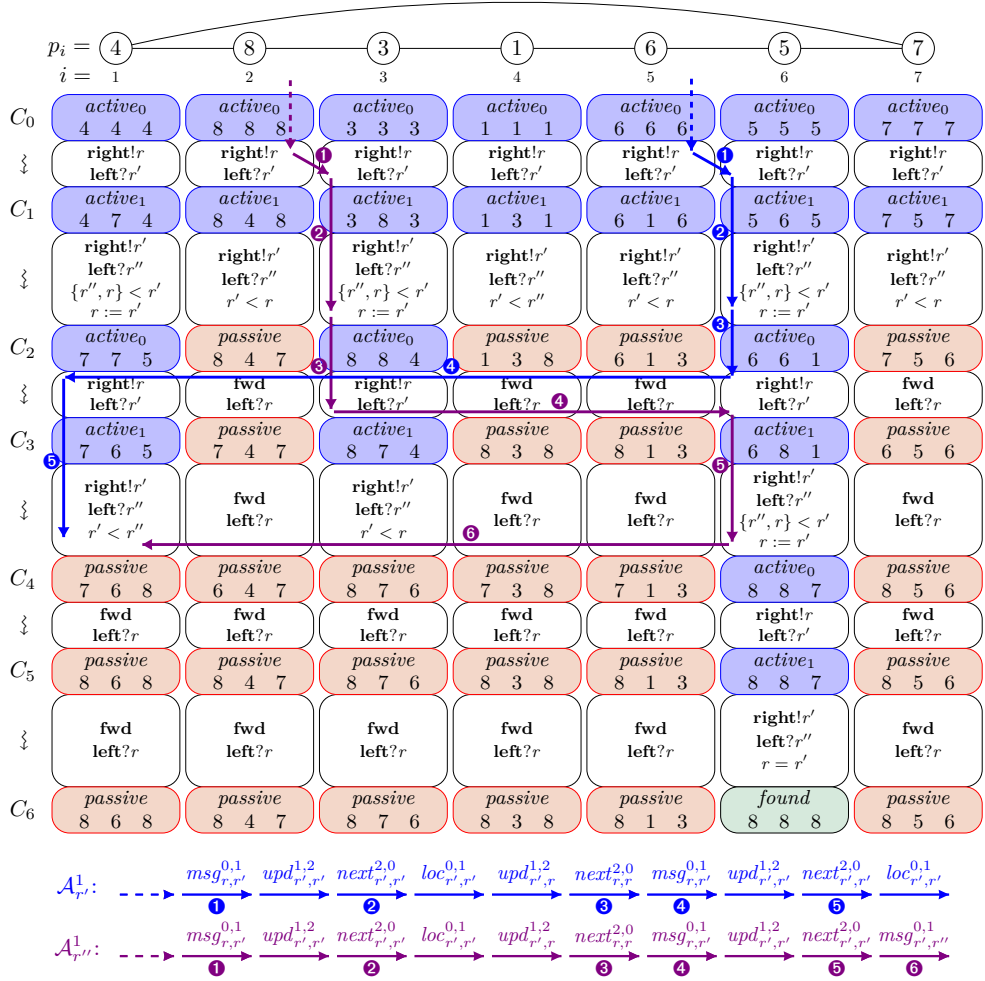

\centering
\scalebox{0.85}{
\gusepicture{symbrun}
}
\caption{Run of Dolev-Klawe-Rodeh algorithm and runs of path automata\label{fig:symbrun}}
\end{figure}

The guards in the transitions $t_1,\ldots,t_n$ are checked against ``intermediate'' register assignments $\hat{\regmap}_1,\ldots,\hat{\regmap}_n: \Reg \to \{p_1,\ldots,p_n\}$, which are defined as follows:
$$\hat{\regmap}_j(r') =
\begin{cases}
\regmap_i(r) & \text{ if } \auxright{r}{i}{r'}{j} \text{ or } \auxleft{r}{i}{r'}{j} \\
\regmap_j(r') & \text{ if, for all } r,i\text{, neither } \auxright{r}{i}{r'}{j} \text{ nor } \auxleft{r}{i}{r'}{j}
\end{cases}$$
Note that this is well-defined, due to condition (1) in Definition~\ref{def:da}.

Now, we write $\conf \confrel{t} \conf'$ if, for all $j \in \set{n}$ and $r,r' \in \Reg$, the following hold:
\begin{enumerate}\itemsep=1ex
\item $s_j \in t_j$ and $(\gotocmd{s_j'}) \in t_j$,
\item $\hat{\regmap}_j(r) < \hat{\regmap}_j(r')$ ~~if $(r<r') \in {t_j}$,
\item $\hat{\regmap}_j(r) = \hat{\regmap}_j(r')$ ~~if $(r=r') \in {t_j}$,
\item $\regmap_j'(r) =
\begin{cases}
\hat{\regmap}_j(r') & \text{ if } (\updcmd{r}{r'}) \in t_j \\
\hat{\regmap}_j(r) & \text{ if } t_j \text{ does not contain an update of the form } \updcmd{r}{r''} \\
\end{cases}$
\end{enumerate}
Again, 4.\ is well-defined thanks to condition (2) in Definition~\ref{def:da}.


An ($\ring$-)\emph{run} of $\DA$ is a sequence $\run = {\conf_0 \confrel{t^1} \conf_1 \confrel{t^2} \ldots \confrel{t^k} \conf_k}$ where $k \ge 1$, $\conf_0$ is the initial $\ring$-configuration,
and $t^j = (t_{1}^j,\ldots,t_{n}^j) \in \Trans^n$ for all $j \in \set{k}$.
We call $k$ the \emph{length} of $\run$. Note that $\run$ uniquely determines the underlying ring $\ring$.

\begin{myremark}
A receive command is always non-blocking even if there is no corresponding send.
As an alternative semantics, one could require that it can only be executed if there has been a matching send, or vice versa. One could even include tags from a finite alphabet that can be sent along with pids. All this will not change any of the forthcoming results.
\myqed
\end{myremark}

\begin{example}
A run of $\DA_{\mathsf{DKR}}$ from Example~\ref{ex:dolev} on the ring $\ring=(7:4,8,3,1,6,5,7)$ is depicted in Figure~\ref{fig:symbrun} (for the moment, we may ignore the blue and violet lines). A colored row forms a configuration. The three pids in a cell refer to registers $r,r',r''$, respectively (we ignore $\idreg$). Moreover, a non-colored row forms, together with the states above and below, a transition tuple.
When looking at the step from $\conf_3$ to $\conf_4$, we have, for example, $\auxright{r'}{3}{r}{4}$ and $\auxright{r'}{3}{r''}{6}$. Moreover, $\auxright{r'}{6}{r}{7}$ and $\auxright{r'}{6}{r''}{1}$ (recall that we are in a ring). Note that the run conforms to the correctness property formulated in Example~\ref{ex:dolev}. In particular, in the final configuration, all processes store the maximum pid in register $r$.
\myqed
\end{example}


\section{The Specification Language}\label{sec:spec}

\newcommand{\existsp}[5]{\langle #3 \rangle#1 #5 \langle #4 \rangle#2}
\newcommand{\forallp}[5]{[#3]#1 #5 [#4]#2}

\newcommand{\existseq}[4]{\existsp{#1}{#2}{#3}{#4}{=}}
\newcommand{\existsleq}[4]{\existsp{#1}{#2}{#3}{#4}{\le}}
\newcommand{\existsless}[4]{\existsp{#1}{#2}{#3}{#4}{<}}
\newcommand{\foralleq}[4]{[#1@#3 = #2@#4]}
\newcommand{\forallleq}[4]{[#1@#3 \le #2@#4]}
\newcommand{\forallless}[4]{[#1@#3 < #2@#4]}
\newcommand{\Allrings}[1]{\forall_{\!\mathit{rings}}\forall_{\!\mathit{runs}}\forall_{\marked}#1}
\newcommand{\Allexistsrings}[1]{\forall_{\!\mathit{rings}}\exists_{\mathit{run}}\forall_{\marked}#1}

\newcommand{\marked}{\mathsf{m}}
\newcommand{\DataPDL}{\ensuremath{\textup{DataPDL}}\xspace}
\newcommand{\DataPDLm}{\ensuremath{\textup{DataPDL}^{\ominus}}\xspace}

In Examples~\ref{ex:franklin} and \ref{ex:dolev}, we informally stated the correctness criterion for the presented algorithms (e.g., ``at the end, all processes store the maximal pid in register $r$'').
Now, we introduce a \emph{formal} language to specify correctness properties. It is defined wrt.\ a given distributed algorithm $\DA=(\States,\init,\Reg,\Trans)$, which we fix for the rest of this section.

Typically, one requires that a distributed algorithm is correct no matter what the underlying ring is. Since we will bound the number of rounds, we moreover study a form of partial correctness.
Accordingly, a property is of the form $\Allrings{\locform}$, which has to be read as ``for all rings, all runs, and all processes $\marked$, we have $\locform$''. The marking $\marked$ is used to avoid to ``get lost'' in a ring when writing the property $\locform$. This is like placing a pebble in the ring that can be retrieved at any time. Actually, $\locform$ allows us to ``navigate'' back and forth ($\goup$ and $\godown$) in a run, i.e., from one configuration to the previous or next one (similar to a temporal logic with past operators). By means of $\goleft$ and $\goright$, we may also navigate horizontally within a configuration, i.e., from one process to a neighboring one.

Essentially, a sequence of configurations is interpreted as a cylinder (cf.\
Figure~\ref{fig:symbrun}) that can be explored using regular expressions $\pi$
over $\{\stay,\goleft,\goright,\goup,\godown\}$ (where $\stay$ means ``stay'').
At a given position/coordinate of the cylinder, we can check \emph{local (or positional)} properties like
the state taken by a process, or whether we are on the marked process
$\marked$. Such a property can be combined with a regular expression $\pi$:
The formula $\forallpath{\pi}\locform$ says that $\locform$ holds at every position
that is reachable through a $\pi$-path (a path matching $\pi$).
Dually, $\existspath{\pi}\locform$ holds if there is a $\pi$-path to some position where $\locform$ is satisfied.
The most interesting construct in our logic is $\existsp{r}{r'}{\pi}{\pi'}{\bowtie}$, where ${\bowtie} \in
\{\mathord{=},\mathord{\neq},\mathord{<},\mathord{\le}\}$, which has been
used for reasoning about XML documents \cite{BenediktFG08,BojanczykMSS09,FS11}.
It says that, from the current position, there are a $\pi$-path and a $\pi'$-path
that lead to positions $y$ and $y'$, respectively, such that the pid stored in
register $r$ at $y$ and the pid stored in $r'$ at $y'$ satisfy the relation
$\bowtie$.

We will now introduce our logic in full generality. Later, we will restrict the use of $<$- and $\le$-guards to obtain positive results.

\begin{definition}\label{def:datapdl}
The logic $\DataPDL(\DA)$ is given by the following grammar:
\begin{align*}
\Phi &::= \Allrings{\locform}\\
\locform,\locform' &::= \marked \,\mid\, s \,\mid\, \neg\locform \,\mid\, \locform \wedge \locform' \,\mid\, \locform \Rightarrow \locform' \,\mid\, \forallpath{\pi}\locform \,\mid\, \existsp{r}{r'}{\pi}{\pi'}{\bowtie}\\
  \pi,\pi' &::= \test{\locform} \,\mid\, d \,\mid\, \pi + \pi' \,\mid\, \pi \cdot \pi' \,\mid\, \pi^{\ast}
\end{align*}
where $s \in \States$, $r,r' \in \Reg$, ${\bowtie} \in \{\mathord{=},\mathord{\neq},\mathord{<},\mathord{\le}\}$, and $d \in \{\stay,\goleft,\goright,\goup,\godown\}$.
\myqed
\end{definition}

We call $\locform$ a \emph{local formula}, and $\pi$ a \emph{path formula}. We use common abbreviations such as $\false = \marked \wedge \neg\marked$, $\existspath{\pi}\locform = \neg\forallpath{\pi}\neg\phi$, and $\locform \vee \locform' = \neg(\neg\locform \wedge \neg\locform')$, and we may write $\pi\pi'$ instead of $\pi \cdot \pi'$.
Implication $\Rightarrow$ is included explicitly in view of the restriction defined below.

Next, we define the semantics. 
Consider a run $\run = {\conf_0 \confrel{t^1} \conf_1 \confrel{t^2} \ldots \confrel{t^k} \conf_k}$ of $\DA$ where $\conf_j = (s_1^j,\ldots,s_n^j,\regmap_1^j,\ldots,\regmap_n^j)$, i.e., $n$ is the number of processes in the underlying ring.
A local formula $\locform$ is interpreted over $\run$ wrt.\ a marked process $m \in \set{n}$ and a position $(i,j) \in \Coord{\run}$ where $\Coord{\run} = \set{n} \times \setz{k}$.
Let us define when $\run,m,(i,j) \models \locform$ holds.
The operators $\neg$, $\wedge$, and $\Rightarrow$ are as usual. Moreover, $\run,m,(i,j) \models \marked$ if $i = m$, and $\run,m,(i,j) \models s$ if $s_i^j = s$.

The other local formulas use path formulas. The semantics of a path formula $\pi$ is given in terms of a binary relation $\Sem{\pi}{\run,m} \subseteq \Coord{\run} \times \Coord{\run}$, which we define below. First, we set:
\begin{itemize}\itemsep=1ex
\item $\run,m,(i,j) \models \forallpath{\pi}\locform$ if $\forall (i',j')$ such that $((i,j),(i',j')) \in \Sem{\pi}{\run,m}$, we have $\run,m,(i',j') \models \locform$
\item $\run,m,(i,j) \models \existsp{r}{r'}{\pi}{\pi'}{\bowtie}$ (where ${\bowtie} \in \{=,\neq,\mathord{<},\mathord{\le}\}$) if $\exists (i_1,j_1),(i_2,j_2)$ such that $((i,j),(i_1,j_1)) \in \Sem{\pi}{\run,m}$ and $((i,j),(i_2,j_2)) \in \Sem{\pi'}{\run,m}$ and $\regmap_{i_1}^{j_1}(r) \bowtie \regmap_{i_2}^{j_2}(r')$
\end{itemize}

It remains to define $\Sem{\pi}{\run,m}$
for a path formula $\pi$.  First, a local test and a stay $\stay$ do not
``move'' at all: $\Sem{\test{\locform}}{\run,m} = \{(x,x) \mid x \in
\Coord{\run}$ such that $\run,m,x \models \locform\}$, and $\Sem{\stay}{\run,m}
= \{(x,x) \mid x \in \Coord{\run}\}$.  Using $\goright$, we
move to the right neighbor of a process: $\Sem{\goright}{\run,m} =
\{((i,j),(i+1,j)) \mid i \in \set{n-1}$ and $j \in \setz{k}\} \cup
\{((n,j),(1,j)) \mid j \in \setz{k}\}$.  We define $\Sem{\goleft}{\run,m}$
accordingly.  Moreover, $\Sem{\godown}{\run,m} = \{((i,j),(i,j+1)) \mid i \in
\set{n}$ and $j \in \setz{k-1}\}$, and similarly for
$\Sem{\goup}{\run,m}$.  The regular constructs, $+$, $\cdot$, and $\ast$ are
as expected and refer to the union, relation composition, and star over binary
relations.

Finally, $\DA$ satisfies the \DataPDL formula $\Allrings{\locform}$, written $\DA \models \Allrings{\locform}$, if, for all rings $\ring=(n:\ldots)$, all $\ring$-runs $\run$,
and all processes $m \in \set{n}$, we have $\run,m,(m,0) \models \locform$. Thus, $\locform$ is evaluated at the first configuration, wrt.\ all processes $m$.

Next, we define a restricted logic, $\DataPDLm(\DA)$, for which we later present our main result.
We say that a path formula $\pi$ is \emph{unambiguous} if, from a given position, it defines at most one reference point. Formally, for all rings $\ring=(n:\ldots)$, $\ring$-runs $\run$ of $\DA$, processes $m \in \set{n}$, and positions $x \in \Coord{\run}$, there is at most one $x' \in \Coord{\run}$ such that $(x,x') \in \Sem{\pi}{\run,m}$. For example, $\stay$, $\godown$, $\goright$, and $\goright^\ast\test{\marked}$ are unambiguous, while $\goright^\ast$ and $\goleft + \goright$ are not unambiguous.

\begin{definition}
A $\DataPDL(\DA)$ formula is contained in $\DataPDLm(\DA)$ if every subformula $\phi=\existsp{r}{r'}{\pi}{\pi'}{\bowtie}$ with ${\bowtie} \in \{<,\le\}$ is such that $\pi$ and $\pi'$ are \emph{unambiguous}. Moreover, $\phi$ must \emph{not} occur (i) in the scope of a negation, (ii) on the left-hand side of an implication $\underline{~\;} \!\Rightarrow\! \underline{~\;}\,$, or (iii) within a test $\test{\,\underline{~\;}\,}$. Note that guards using $=$ and $\neq$ are still unrestricted.
\myqed
\end{definition}

\newcommand{\nextfound}{\pi_\mathsf{found}}

\begin{example}\label{ex:datapdl}
Let us \emph{formalize}, in $\DataPDLm(\DA)$, the correctness criteria for $\DA_\mathsf{Franklin}$ and $\DA_\mathsf{DKR}$ that we stated informally in Examples~\ref{ex:franklin} and \ref{ex:dolev}.
Consider the following local formulas:
\[\begin{array}{ll}
\locform_\mathsf{last}=\forallpath{\godown}\false &
\locform_\mathsf{max} = \forallpath{\goright^\ast} \bigl(\existsleq{\idreg}{r}{\stay}{\nextfound}\bigr)\\[1ex]

\locform_\mathsf{acc} = \forallpath{\goright^\ast}(\passiv \vee \found) &
\locform_{r=\idreg} = \existspath{\nextfound}\bigl(\existsp{r}{\idreg}{\stay}{\stay}{=}\bigr)\\[1ex]

\locform_\mathsf{found} = \existspath{\nextfound\goright(\test{\neg\found}\goright)^\ast}\marked ~~~~~~ &
\locform_{r=r} = \neg\bigl(\existsp{r}{r}{\stay}{\goright^\ast}{\neq}\bigr)
\end{array}\]
where $\nextfound=(\test{\neg\found}\goright)^\ast\test{\found}$.
Note that $\nextfound$
is unambiguous: while going to
the right, it always stops at the \emph{nearest} process that is in state
$\found$. Thus, $\phi_\mathsf{max}$ is indeed a local \DataPDLm formula.
Consider the $\DataPDLm$ formula
\[\Phi_1 = \Allrings{
\forallpath{\godown^\ast}\bigl((\locform_\mathsf{last}\wedge\locform_\mathsf{acc}) \Rightarrow (\locform_\mathsf{found} \wedge \locform_\mathsf{max} \wedge \locform_{r=r} \wedge \locform_{r=\idreg})\bigr)
}\,.\]
It says that, at the end (i.e., in the last configuration) of each accepting run, expressed by $\forallpath{\godown^\ast}\bigl((\locform_\mathsf{last}\wedge\locform_\mathsf{acc}) \Rightarrow{\ldots}\bigr)$, we have that
\begin{itemize}\itemsep=0.5ex
\item[(i)] there is exactly one process $i_0$ that ends in state $\found$ (guaranteed by $\locform_\mathsf{found}$),
\item[(ii)] register $r$ of $i_0$ contains the maximum over all pids ($\locform_\mathsf{max}$),
\item[(iii)] register $r$ of $i_0$ contains the pid of $i_0$ itself ($\locform_{r=\idreg}$), and
\item[(iv)] all processes store the same pid in $r$ ($\locform_{r=r}$).
\end{itemize}
Thus, $\DA_\mathsf{Franklin} \models \Phi_1$. On the other hand, we have $\DA_\mathsf{DKR} \not\models \Phi_1$, because in $\DA_\mathsf{DKR}$ the process that ends in $\found$ is not necessarily the process with the maximum pid. However, we still have $\DA_\mathsf{DKR} \models
\Phi_2$ where \[\Phi_2=\Allrings{
\forallpath{\godown^\ast}\bigl((\locform_\mathsf{last}\wedge\locform_\mathsf{acc}) \Rightarrow (\locform_\mathsf{found} \wedge \locform_\mathsf{max} \wedge \locform_{r=r})\bigr)
}\,.\]

The next example formulates the correctness constraint for a distributed sorting algorithm. We would like to say that, at the end of an accepting run, the pids stored in registers $\reg$ are strictly totally ordered.
Suppose $\locform_\mathsf{acc}$ represents an acceptance condition and $\locform_\mathsf{least}$ says that there is exactly one process that terminates in some dedicated state $\mathit{least}$, similarly to $\locform_\mathsf{found}$ above. Then, \[\Phi_3 = \Allrings{
\forallpath{\godown^\ast}\bigl((\locform_\mathsf{last}\wedge\locform_\mathsf{acc}) \Rightarrow 
(\phi_\mathsf{least} \wedge \forallpath{\goright^\ast\test{\neg\mathit{least}}}
(\existsless{r}{r}{\goleft}{\stay}))\bigr)}\] makes sure that, whenever process $j$ is 
not terminating in $\mathit{least}$, its left
neighbor $i$ stores a smaller pid in $r$ than $j$ does.

Note that $\Phi_1$, $\Phi_2$, and $\Phi_3$ are indeed \DataPDLm formulas.
\myqed
\end{example}

Unsurprisingly, model checking distributed algorithms against \DataPDLm is undecidable:

\begin{theorem}\label{thm:undecidable}
The following problem is undecidable: Given a distributed algorithm $\DA$ and $\Phi \in \DataPDLm(\DA)$, do we have $\DA \models \Phi$\,? (Actually, this even holds for formulas $\Phi$ that express simple state-reachability properties and do not use any guards on pids.)
\end{theorem}


\section{Round-Bounded Model Checking}
\label{sec:verification}

In the realm of multithreaded concurrent programs, where model checking is undecidable in general, a fruitful approach has been to underapproximate the behavior of a system \cite{Qadeer:TACAS05}. The idea is to introduce a parameter that measures a characteristic of a run such as the number of thread switches it performs. One then imposes a bound on this parameter and explores all behaviors up to that bound. In numerous distributed algorithms, the number $\bound$ of rounds needed to conclude is exponentially smaller than the number of processes (cf.\ Examples~\ref{ex:franklin} and \ref{ex:dolev}). Therefore, $\bound$ seems to be a promising parameter for bounded model checking of distributed algorithms.

For a distributed algorithm $\DA$, a formula $\Phi=\Allrings{\locform} \in \DataPDL(\DA)$, and  $\bound \ge 1$, we write $\DA \models_\bound \Phi$ if, for all rings $\ring=(n:\ldots)$, all $\ring$-runs $\run$ of length $k \le \bound$, and all processes $m \in \set{n}$, we have $\run,m,(m,0) \models \locform$. We now present our main result:

\begin{theorem}\label{thm:main}
The following problem is PSPACE-complete: Given a distributed algorithm $\DA$, $\Phi \in \DataPDLm(\DA)$, and a natural number $\bound \ge 1$ (encoded in unary), do we have $\DA \models_\bound \Phi$\,?
\end{theorem}

The lower-bound proof, a reduction from the intersection-emptiness problem for a list of finite automata, can be found in the appendix. Before we prove the upper bound, let us discuss the result in more detail. We will first compare it with ``na{\"i}ve'' approaches to solve related questions. Consider the problem to determine whether a distributed algorithm satisfies its specification for all rings up to size $n$ and all runs up to length $\bound$. This problem is in coNP: We guess a ring (i.e., essentially, a permutation of pids) and a run, and we check, using \cite{Lange06}, whether the run does \emph{not} satisfy the formula. Next, suppose only $\bound$ is given and the question is whether, for all rings up to size $2^\bound$ and all runs up to length $\bound$, the property holds. Then, the above procedure gives us a coNEXPTIME algorithm.

Thus, our result is interesting complexity-wise, but it offers some other advantages. First, it actually checks correctness (up to round number $\bound$) for \emph{all} rings. This is essential when verifying distributed \emph{protocols} against safety properties.
Second, it reduces to a satisfiability check in the well-studied propositional dynamic logic with loop and converse (LCPDL) \cite{Goeller2009}, which in turn can be reduced to an emptiness check of alternating two-way automata (A2As) over words \cite{Vardi1998}. The ``na{\"i}ve'' approaches, on the other hand, do not seem to give rise to viable algorithms. Finally, our approach is uniform in the following sense: We will construct, in polynomial time, an A2A that recognizes precisely the symbolic abstractions of runs (over arbitrary rings) that violate (or satisfy) a given formula. Our construction is \emph{independent} of the parameter $\bound$. The emptiness check then requires a bound on the number of rounds (or on the number of processes), which can be adjusted gradually without changing the automaton.

\paragraph{Proof Outline for Upper Bound of Theorem~\ref{thm:main}.}

\newcommand{\Pictures}{\Trans^{++}}
\newcommand{\pict}{T}
\newcommand{\RunsTR}[2]{\mathit{Runs}_{#1,#2}}
\newcommand{\Runs}[1]{\mathit{Runs}(#1)}
\newcommand{\posRuns}[1]{\mathit{Runs}^+(#1)}
\newcommand{\negRuns}[1]{\mathit{Runs}^-(#1)}

Let $\DA$ be the given distributed algorithm and $\Phi \in \DataPDLm(\DA)$.
We will reduce model checking to the satisfiability problem for LCPDL \cite{Goeller2009}. While \DataPDLm is interpreted over runs, containing pids from an infinite alphabet, the new logic will reason about symbolic abstractions over a \emph{finite} alphabet. A symbolic abstraction of a run only keeps the transitions and discards pids. Thus, it can be seen as a table (or picture) whose entries are transitions (cf.\ Figure~\ref{fig:symbrun}).

First, we translate $\DA$ into an LCPDL formula. Essentially, it checks that guards are not used in a contradictory way.
To compare $\DA$ with $\Phi$, the latter is translated into an LCPDL formula, too. However, there is a subtle point here.
For simplicity, let us write $r < r'$ instead of $\existsp{r}{r'}{\stay}{\stay}{<}$.
Satisfaction of a formula $r < r'$ can only be guaranteed in a symbolic execution if the flow of pids provides \emph{evidence} that $r < r'$ really holds. More concretely, the (hypothetic) formula $(r < r') \vee (r = r') \vee (r' < r)$ is a tautology, but it may not be possible to prove any of its disjuncts on the basis of a symbolic run.
This is the reason why $\DataPDLm$ restricts $<$- and $\le$-tests. It is then indeed enough to reason about symbolic runs (cf.\ Lemma~\ref{lem:lcpdl} below).
We leave open whether one can deal with full \DataPDL.

Overall, we reduce model checking to satisfiability of the conjunction of two \ICPDL formulas of polynomial size: the formula representing the algorithm, and the negation of the formula representing the specification. Satisfiability of LCPDL over symbolic runs (of bounded height) can be checked in PSPACE \cite{Goeller2009} by a reduction to the emptiness problem for A2As over words \cite{Vardi1998}. Our approach is, thus, automata theoretic in spirit, though the power of alternation is used differently than in \cite{Vardi1996}, which translates LTL formulas into automata.

Next, we present the logic LCPDL over symbolic runs. Then, in separate subsections, we translate $\DA$ as well as its \DataPDLm specification into LCPDL. For the remainder of this section, we fix a distributed algorithm $\DA=(\States,\init,\Reg,\Trans)$.

\paragraph{PDL with Loop and Converse (LCPDL).}

As mentioned before, a symbolic abstraction of a run of $\DA$ is a table, whose entries are transitions from the finite alphabet $\Trans$. A \emph{table} is a triple $\pict=(n,k,\lambda)$ where $n,k \ge 1$ and $\lambda: \Coord{T} \to \Trans$ labels each position/coordinate from $\Coord{T}= \set{n} \times \setz{k}$ with a transition. Thus, we may consider that $T$ has $n$ columns and $k+1$ rows. In the following, we will write $\coord{\pict}{i}{j}$ for $\lambda(i,j)$, and $\col{\pict}{i}$ for the $i$-th column of $T$, i.e., $\col{\pict}{i}=\pict[i,0] \ldots \pict[i,k] \in \Trans^+$.
Let $\Pictures$ denote the set of all tables.

\newcommand{\roundright}{\rotatebox[origin=cc]{180}{$\circlearrowleft$}}
\newcommand{\roundleft}{\rotatebox[origin=cc]{180}{$\circlearrowright$}}

Formulas $\lcpdl \in \ICPDL(\DA)$ are interpreted over tables. Their syntax is given as follows:
\[\begin{array}{rcl}
  \lcpdl,\lcpdl' &\!\!::=\!\!& t \mid s \mid \gotocmd{s} \mid \fwdcmd \mid \sendleft{r} \mid \sendright{r} \mid \recleft{r} \mid \recright{r} \mid r<r' \mid r=r' \mid \updcmd{r}{r'} \mid\\[0.5ex]
  & & \neg \lcpdl \mid \lcpdl \wedge \lcpdl' \mid \Existspath{\pi}{\lcpdl} \mid \loopform{\pi}
  \\[1ex]
  \pi,\pi' &\!\!::=\!\!& \test{\lcpdl} \mid d \mid \pi + \pi' \mid \pi \cdot \pi' \mid \pi^{\ast} \mid \pi^{-1} \mid \pathaut
\end{array}\]
where $t \in \Trans$, $s \in \States$, $r,r' \in \Reg$, $d \in \{\stay,\goright,\godown\}$, and $\pathaut$ is a \emph{path automaton}: a non-deterministic finite automaton whose transitions are labeled with path formulas $\pi$. Again, $\lcpdl$ is called a \emph{local formula}.
We use common abbreviations to include disjunction, implication, $\true$, and $\false$, and we let $\pi^+ = \pi \cdot \pi^\ast$, $\forallpath{\pi}\lcpdl \formdef \neg\existspath{\pi}\neg\lcpdl$, $\existspath{\pi} \formdef \existspath{\pi}\true$, $\goleft = \goright^{-1}$, and $\goup = \godown^{-1}$.

The semantics of \ICPDL is very similar to that of \DataPDL. A local formula $\lcpdl$ is interpreted over a table $\pict=(n,k,\lambda)$ and a position $x \in \Coord{T}$. When it is satisfied, we write $\pict,x \models \lcpdl$. Moreover, a path formula $\pi$ determines a binary relation $\Sem{\pi}{\pict} \subseteq \Coord{T} \times \Coord{T}$, relating those positions that are connected by a path matching $\pi$.

We consider only the most important cases: We have $\pict,(i,j) \models t$ if
$\pict[i,j]=t$.  For a state, command, guard, or update $\gamma$, let
$\pict,(i,j) \models \gamma$ if $\gamma \in \pict[i,j]$.
Loop and converse are as expected:
$\pict,x \models \loopform{\pi}$ if $(x,x) \in \Sem{\pi}{\pict}$, and
$\Sem{\pi^{-1}}{\pict} = \{(y,x) \mid (x,y) \in \Sem{\pi}{\pict}\}$.
The semantics of
$\goright$ (and $\goleft$) is slightly different than in \DataPDL, since we are not allowed to go beyond the last and first column.
Thus, $\Sem{\goright}{\pict} =
\{((i,j),(i+1,j)) \mid i \in \set{n-1}$ and $j \in \setz{k}\}$.
However, we can simulate the ``roundabout'' of a ring and set ${\hookrightarrow} = \goright + \test{\neg\existspath{\goright}}\goleft^\ast\test{\neg\existspath{\goleft}}$ as well as ${\hookleftarrow} = {\hookrightarrow^{-1}}$.
Actually, the
first column of a table will play the role of a marked process in a ring (later, $\marked$ will be
translated to $\neg\existspath{\goleft}$).

Finally, the semantics of path automata is given by $\Sem{\A}{\pict} = \{(x,y) \mid$ there is $\pi_1 \ldots \pi_\ell \in L(\A)$ with $(x,y) \in \Sem{\pi_1 \cdot \ldots \cdot \pi_\ell}{\pict}\}$ where $L(\A)$ contains a \emph{sequence} $\pi_1 \ldots \pi_\ell$ of path formulas if $\A$ admits a path $q_0 \xrightarrow{\pi_1} q_1 \xrightarrow{\pi_2} \ldots \xrightarrow{\pi_\ell} q_\ell$ from its initial state $q_0$ to a final state $q_\ell$.

A formula $\psi \in \ICPDL(\DA)$ defines the language $L(\lcpdl) = \{\pict \in \Pictures \mid \pict,(1,0) \models \lcpdl\}$. For $\bound \ge 1$, we denote by $L_\bound(\lcpdl)$ the set of tables $(n,k,\lambda) \in L(\lcpdl)$ such that $k \le \bound$.

\begin{theorem}[essentially \cite{Goeller2009}]\label{thm:icpdl}
The following problem is PSPACE-complete: Given a distributed algorithm $\DA$, a formula $\lcpdl \in \ICPDL(\DA)$, and $\bound \ge 1$ (encoded in unary), do we have $L_\bound(\lcpdl) = \emptyset$\,? (The input $\DA$ is only needed to determine the signature of the logic.)
\end{theorem}

\paragraph{From Distributed Algorithms to \ICPDL.}

\begin{figure}[t]
\centering
\begin{tabular}{l}
$\localform{r}{r'} \formdef
\begin{cases}
\test{\bigwedge_{\bar{r} \in \Reg}\!\neg\existspath{(\msgform{\bar{r}}{r})^{-1}}} & \textup{if~} r = r'\\[0.5ex]
\test{\false} & \textup{if~} r \neq r'
\end{cases}$
\qquad
$\updform{r}{r'} \formdef
\begin{cases}
\test{ \bigwedge_{\bar{r} \neq r} \neg(\updcmd{r}{\bar{r}})
} & \textup{if~} r = r'\\[0.5ex]
\test{\updcmd{r'}{r}} & \textup{if~} r \neq r'
\end{cases}$\\[5ex]

$\msgform{r}{r'} \formdef
\left(
\begin{array}{rl}
& \test{\sendright{r}} \cdot (\hookrightarrow \cdot \test{\passiveform})^\ast \cdot \hookrightarrow \cdot\test{\recleft{{r'}}}\\[0.5ex]
\!+ \!\!\!\!& \test{\sendleft{r}} \cdot (\hookleftarrow \cdot \test{\passiveform})^\ast \cdot \hookleftarrow \cdot\test{\recright{{r'}}}
\end{array}\right)$
\qquad
$\nextform{r}{r'} \formdef
\begin{cases}
\godown & \textup{if~} r = r'\\[0.5ex]
\test{\false} & \textup{if~} r \neq r'
\end{cases}$
\end{tabular}
\caption{Path formulas to trace back transmission of pids\label{fig:pathform}}
\end{figure}

\newcommand{\ph}{h}
\newcommand{\emptytrans}{\mathsf{t}}

Wlog., we assume that $\Trans$ contains $\emptytrans=\datrans{\mathsf{s}}{\skipcmd}{\skipcmd}{\skipcmd}{\skipcmd}{\init}$ where $\mathsf{s} \neq \init$ does not occur in any other transition.

Let $\ring=(n:p_1,\ldots,p_n)$ be a ring and $\run = {\conf_0 \confrel{t^1} \conf_1 \confrel{t^2} \ldots \confrel{t^k} \conf_k}$ be an $\ring$-run of $\DA$, where $t^j = (t_{1}^j,\ldots,t_{n}^j) \in \Trans^n$ for all $j \in \set{k}$. From $\run$, we extract the \emph{symbolic run} $\pictofrun{\run}=(n,k,\lambda) \in \Pictures$ given by its columns $\pictofrun{\run}[i] = \emptytrans\, t_i^1 \ldots t_i^k$. The purpose of the dummy transition $\emptytrans$ at the beginning of a column is to match the number of configurations in a run.

We will construct, in polynomial time, a formula $\lcpdl_\DA \in \ICPDL(\DA)$
such that $L(\lcpdl_\DA) = \{\pictofrun{\run} \mid \run$ is a run of $\DA\}$.
In particular, $\lcpdl_\DA$ will verify that (i) there are no
cyclic dependencies that arise from $<$-guards, and (ii) registers in equality
guards can be traced back to the same origin. In that case, the
symbolic run is consistent and corresponds to a ``real'' run of $\DA$.

The main ingredients of $\lcpdl_\DA$ are some path formulas that describe the transmission of pids in a symbolic run.
They are depicted in Figure~\ref{fig:pathform}. For $\theta \in \{\mathit{loc},\mathit{msg},\mathit{upd},\mathit{next}\}$ and $\ph \in \{0,1,2\}$, the meaning of $(x,y) \in \Sem{\theta_{r,r'}^{\ph,\ph'}}{T}$ is that the pid stored in $r$ at \emph{stage} $\ph$ of position/transition $x$ has been propagated to register $r'$ at stage $\ph'$ of $y$. Here, $\ph=0$ means ``after sending'', $\ph=1$ ``after receiving'', and $\ph=2$ ``after register update''. The interpretation of ``propagated'' depends on $\theta$.
Formula $\localform{r}{r'}$ says that the value of register $r$ is not affected by reception. Similarly, $\smash{\updform{r}{r'}}$ takes care of updates. Formula $\smash{\nextform{r}{r'}}$ allows us to switch to the next transition of a process, preserving the value of $r (= r')$. The most interesting case is $\smash{\msgform{r}{r'}}$, which describes paths across several processes. It relates the sending of $r$ and a corresponding receive in $r'$, which requires that all intermediate transitions are forward transitions. All path formulas are illustrated in Figure~\ref{fig:symbrun}.

Since pids can be transmitted along several transitions and messages, the formulas $\smash{\theta_{r,r'}^{\ph,\ph'}}$ will be composed by path automata. For $\textup{\ph} \in \{1,2\}$ and $\textup{r} \in \Reg$, we define a path automaton $\smash{\pathaut_{\textup{r}}^{\textup{\ph}}}$ that, in $T_\run$, connects some positions $(i,0)$ and $(i',j')$ iff, in $\run$, register $\textup{r}$ stores $p_i$ at stage $\textup{h}$ of position $(i',j')$. Its set of states is $\iota \cup (\{0,1,2\} \times \Reg)$. For all $\reg \in \Reg$, there is a transition from the initial state $\iota$ to $(0,r)$ with transition label $\test{\neg\existspath{\goup}}$. Thus, the automaton starts at the top row and non-deterministically chooses some register $r$. From state $(\ph,r)$, it can read any transition label $\smash{\theta_{r,r'}^{\ph,\ph'}}$ and move to $(\ph',r')$. The only final state is $(\textup{\ph},\textup{r})$. Figure~\ref{fig:symbrun} describes (partial) runs of $\pathaut_{r'}^{1}$ and $\pathaut_{r''}^{1}$, which allow us to identify the origin of $r'$ and $r''$ when applying the guard $r'<r''$.

Now, consistency of equality guards can indeed be verified by an \ICPDL formula. It says that, whenever an equality check $r = r'$ occurs in the symbolic run, then the pids stored in $r$ and $r'$ have a common origin. This can be conveniently expressed in terms of loop and converse. Note that guards are checked at stage $\ph=1$ of the corresponding transition:
\[\lcpdl_= ~\formdef~ \forallpath{(\goright + \godown)^\ast} \textstyle\bigwedge_{r,r' \in \Reg} \Bigl(r = r' ~\Rightarrow~ \loopform{(\pathaut_{r}^{1})^{-1} \cdot \pathaut_{r'}^{1}}\Bigr)\,.\]
The next path formula connects the first coordinate of a process $i$ with the first coordinate of another process $i'$ if some guard forces the pid of $i$ to be smaller than that of $i'$:
\[\pi_< ~\formdef \Bigl(\textstyle\sum_{r,r' \in \Reg} \pathaut_{r}^{1} \cdot \test{r < r'} \cdot (\pathaut_{r'}^{1})^{-1}\Bigr)^+\,.\]
Note that, here, we use the (strict) transitive closure. Consistency of $<$-guards now reduces to saying that there is no $\pi_<$-loop:
$\lcpdl_< ~\formdef~ \neg\existspath{\goright^\ast} \loopform{\pi_<}$.

Finally, we can easily write an \ICPDL formula $\lcpdl_{\mathsf{col}}$ that checks whether every column $\pict[i] \in \Trans^+$ (ignoring $\emptytrans$) is a valid transition sequence of $\DA$.
Finally, let $\lcpdl_\DA \formdef \lcpdl_= \wedge \lcpdl_< \wedge \lcpdl_{\mathsf{col}}$.

\begin{lemma}\label{lem:datapdl}
We have $L(\lcpdl_\DA) = \{\pictofrun{\run} \mid \run$ is a run of $\DA\}$.
\end{lemma}

\paragraph{From \DataPDLm to \ICPDL.}

\newcommand{\transl}[1]{\widetilde{#1}}

Next, we inductively translate every local $\DataPDLm(\DA)$ formula $\locform$
into an $\ICPDL(\DA)$ formula $\transl{\locform}$.  The translation is given in
Figure~\ref{fig:transl}.  As mentioned before, the first column in a table plays
the role of a marked process so that $\transl{\marked} =
\neg\existspath{\goleft}$.  The standard formulas are translated as expected.
Now, consider $\transl{~\smash{\existsp{r}{r'}{\pi}{\pi'}{<}}}$ (the remaining cases are
similar).  To ``prove'' $\existsp{r}{r'}{\pi}{\pi'}{<}$ at a given position in a
symbolic run, we require that there are a $\transl{\pi}$-path and a
$\transl{\pi}'$-path to coordinates $x$ and $x'$, respectively, whose registers
$r$ and $r'$ satisfy $r < r'$. To guarantee the latter, the pids stored in $r$ and
$r'$ have to go back to coordinates that are connected by a $\pi_<$-path.
Again, using converse, this can be expressed as a loop (cf.\
Figure~\ref{fig:existsless}).  Note that, hereby, $\pathaut_{r}^{2}$
and $\pathaut_{r'}^{2}$ refer to stage $\ph=2$, which reflects the fact
that $\DataPDL$ speaks about \emph{configurations} (determined after updates).
\begin{figure}[t]
\centering
{
\parbox[b]{0.65\textwidth}{
\centering
\begin{tabular}{l}
$\transl{\marked} = \neg\existspath{\goleft}$
\qquad
$\transl{s} = \gotocmd{s}$ ~for all $s \in \States$\\[0.5ex]

$\transl{\neg\locform} = \neg\transl{\locform}$
\quad
$\transl{\locform_1 \wedge \locform_2} = \transl{\locform_1} \wedge \transl{\locform_2}$
\quad
$\transl{\locform_1 \Rightarrow \locform_2} = \transl{\locform_1} \Rightarrow \transl{\locform_2}$
\quad
$\transl{\forallpath{\pi}\locform} = \forallpath{\transl{\pi}}\transl{\locform}$\\[0.5ex]

$\transl{\existsp{r}{r'}{\pi}{\pi'}{<}} = \loopform{\transl{\pi} \cdot (\pathaut_{r}^{2})^{-1} \cdot \pi_< \cdot \pathaut_{r'}^{2} \cdot (\transl{\pi}')^{-1}}$\\[0.5ex]

$\transl{\existsp{r}{r'}{\pi}{\pi'}{\le}} = \loopform{\transl{\pi} \cdot (\pathaut_{r}^{2})^{-1} \cdot (\pi_< + \stay)\cdot \pathaut_{r'}^{2} \cdot (\transl{\pi}')^{-1}}$\\[0.5ex]

$\transl{\existsp{r}{r'}{\pi}{\pi'}{=}} = \loopform{\transl{\pi} \cdot (\pathaut_{r}^{2})^{-1} \cdot \pathaut_{r'}^{2} \cdot (\transl{\pi}')^{-1}}$\\[0.5ex]

$\transl{\existsp{r}{r'}{\pi}{\pi'}{\neq}} = \loopform{\transl{\pi} \cdot (\pathaut_{r}^{2})^{-1} \cdot (\goleft^+ + \goright^+) \cdot \pathaut_{r'}^{2} \cdot (\transl{\pi}')^{-1}}$\\[0.5ex]

$\transl{\pi}$ is inductively obtained from $\pi$ by replacing tests $\test{\locform}$ by $\test{\transl{\locform}}$,\\[-0.3ex]
\qquad $\goright$ by $\hookrightarrow$, and $\goleft$ by $\hookleftarrow$
\end{tabular}
\caption{From \DataPDLm to LCPDL\label{fig:transl}}
}}
~~~
{
\parbox[b]{0.28\textwidth}{
\centering
\scalebox{0.9}{
\begin{gpicture}
\gasset{Nframe=y,Nh=1.4,Nw=1.4,AHLength=1.8,AHlength=1.5}
\unitlength=0.7mm
\node(current)(0,0){}
\node(r)(-20,20){}
\node(rp)(20,20){}
\node(p1)(-10,40){}
\node(p2)(10,40){}
\drawedge[curvedepth=0,AHnb=1](current,r){$\transl{\pi}$}
\drawedge[ELside=r,ELpos=60,curvedepth=-5,AHnb=1](current,rp){$\transl{\pi}'$}
\drawedge[ELside=r,ELpos=50,ELdist=-1.2,curvedepth=0,AHnb=1](rp,current){$(\transl{\pi}')^{-1}$}
\drawedge[ELside=r,ELpos=50,curvedepth=-5,AHnb=1](p1,r){$\pathaut_{r}^{2}$}
\drawedge[ELside=l,ELpos=50,curvedepth=0,AHnb=1](p2,rp){$\pathaut_{r'}^{2}$}
\drawedge[ELside=r,ELpos=60,ELdist=0.2,curvedepth=0,AHnb=1](r,p1){$(\pathaut_{r}^{2})^{-1}$}
\drawedge[ELside=l,ELpos=50,curvedepth=0,AHnb=1](p1,p2){$\pi_<$}
\end{gpicture}
}
\caption{$\transl{\existsp{r}{r'}{\pi}{\pi'}{<}}$\label{fig:existsless}}
}}
\end{figure}

\begin{lemma}\label{lem:lcpdl}
Let $\pict \in \{\pictofrun{\run} \mid \run$ is a run of $\DA\}$ and $\phi$ be a local $\DataPDLm(\DA)$ formula. We have $T,(1,0) \models \transl{\phi} \;\Longleftrightarrow\, \bigl(\run,1,(1,0) \models \phi \text{ for all runs } \run \text{ of } \DA \text{ such that } T_\run = T\bigr)$.
\end{lemma}

Using Lemmas~\ref{lem:datapdl} and \ref{lem:lcpdl}, we can now prove 
Lemma~\ref{lem:satisf} below. Together with Theorem~\ref{thm:icpdl}, the upper bound of Theorem~\ref{thm:main} follows.

\begin{lemma}\label{lem:satisf}
Let $\DA$ be a distributed algorithm, $\Phi=\Allrings{\locform} \in \DataPDLm(\DA)$, and $\bound \ge 1$. We have (a) $\DA \models \Phi \,\Longleftrightarrow\, L(\lcpdl_\DA \wedge \neg\transl{\locform}) = \emptyset$, and
(b) $\DA \models_\bound \Phi \,\Longleftrightarrow\, L_\bound(\lcpdl_\DA \wedge \neg\transl{\locform}) = \emptyset$.
\end{lemma}

\section{Conclusion}\label{sec:conclusion}

In this paper, we provided a conceptually new approach to the verification of distributed algorithms that is robust against small changes of the model. 

Actually, we made some assumptions that simplify the presentation, but are not crucial to the approach and results. For example, we assumed that an algorithm is synchronous, i.e., there is a global clock that, at every clock tick, triggers a round, in which every process participates. This can be relaxed to handle communication via (bounded) channels. Second, messages are pids, but they could contain message contents from a finite alphabet as well. Though the restriction to the class of rings is crucial for the complexity of our algorithm, the logical framework we developed is largely independent of concrete (ring) architectures. Essentially, we could choose any class of architectures for which \ICPDL is decidable.

We leave open whether round-bounded model checking can deal with full \DataPDL, or with properties of the form $\Allexistsrings{\locform}$, which are branching-time in spirit.




\appendix


\clearpage

\section{Proof of Theorem~\ref{thm:undecidable}}
\newcommand{\rleader}{r_\text{leader}}
\newcommand{\rstate}{r_\text{state}}

The following remark will be exploited in the proof of Theorem~\ref{thm:undecidable} and for the lower-bound proof of Theorem~\ref{thm:main}.

\begin{myremark}\label{rem:messagesAlphabet}
Note that the only way to communicate information from one process to another is by exchanging and comparing pids. However, we can simulate the exchange of messages from a \emph{finite} alphabet $B=\{b_1,\ldots,b_{k}\}$ that can be compared for equality. 

Assume a ring $\ring=(n:p_1,\ldots,p_n)$. A possible protocol for simulation can employ a leader election algorithm first. Afterwards, the leader identifies $k$ distinct pids (say the $k$ closest pids on its left), and transmits them to all other processes who keep them in dedicated registers $\hat{r}_1, \ldots, \hat{r}_k$. After this initialization phase, the actual simulation can take place with the convention that message $b_j$ is identified by the pid in $\hat{r}_j$ (of any process). In order for the simulation to work, we have to require that $n \ge k$.

The drawback of the above protocol is that the initialization phase requires $\log(n)$ rounds. Below we describe another protocol where the initialization can be achieved in $k$ rounds. 

Assume a ring $\ring=(n:p_1,\ldots,p_n)$ and that $n \ge k$. 
Each process has $k+1$ dedicated registers $\hat{r}_0, \ldots, \hat{r}_k$.  After the initialization (described below), for each process $i$, register $\hat{r}_j$ holds $p_{i-j}$ (modulo $n$). Thus $\hat{r}_j$ of process $i$ holds the same value as $\hat{r}_{j+1}$ of process $i+1$. 

\textbf{Conventions.} To send message $b_j$ to left, a process  simply sends the contents of $\hat{r}_j$. On the other hand, to send message $b_j$ to right, it sends the contents of $\hat{r}_{j-1}$.  When a process receives a message from the left, it compares it with registers $\hat{r}_1, \ldots, \hat{r}_k$, and if it matches $\hat{r}_j$ then the message is interpreted as $b_j$. On receiving from right, on contrary, it is compared to $\hat{r}_0, \ldots \hat{r}_{k-1}$,  and if it matches $\hat{r}_j$ then the message is interpreted as $b_{j+1}$.

\textbf{Initialization.} It uses $k+1$ control states $s_0, \ldots, s_k$. At $s_0$, all registers have self pid. This fills in the correct value for $\hat{r}_0$. In round $j$, a process moves from $s_{j-1}$ to $s_j$, sending $\hat{r}_{j-1}$ to the right and receiving in $\hat{r}_j$ from the left.

Notice that this simulation cannot be used to forward a message to another process using $\fwdcmd$-commands in between. However, the lower bound proofs presented below do not rely on $\fwdcmd$-commands.
\myqed
\end{myremark}

\begin{proof}[Proof of Theorem~\ref{thm:undecidable}]
We give a reduction from the halting problem of Turing machines. It is equivalent to checking whether a given Turing machine $\tm$ can never reach a specific target state (call it \halt) on any (some) input. Let $\stm$ be the set of control states of a Turing machine. Let $\btm$ be the tape alphabet of the Turing Machine. Wlog., we assume that the $\tm$ starts on the empty tape. From the empty tape, it may simulate an arbitrary input using non-determinism.  We also assume that, on reaching the state $\halt$,  it writes $\halt$ in the current cell. Thus $\halt \in \stm$ and $\halt \in \btm$.  We describe the distributed algorithm $\Dtm$.

 Intuitively, the number of processes in the ring gives an upper bound to the space needed by the Turing machine. Every process will correspond to a cell in the Turing machine's work tape. Since there is no specific starting process for a ring, we run a leader election algorithm first, and the leader will act as the leftmost cell of the tape. The $i$-th process to the right of the leader acts as the $i$-th tape cell. 
 The local state of processes indicate the corresponding cell contents. It also indicates whether the head is currently present at the respective cell. 
Thus the local states are pairs of the form $(\sym, \head)$ where $\sym \in \btm$ indicates the content of a tape cell, and $\head$ is a boolean value denoting the presence of the head of the Turing machine at the current cell. Initially, only the leader process has the $\head$ bit set $\true$.  In the simulation, only the process with $\head = \true$ can send messages, and once it emits a message, the $\head$ bit is turned $\false$. The process that receives the message turns the $\head$ bit $\true$. The message alphabet (cf. Remark~\ref{rem:messagesAlphabet}) is $\stm $ which denotes the target control state upon simulating one transition of the Turing machine. The control state of the $\tm$ is stored in a designated register $\rstate$.

We describe the construction in detail now. There are two preliminary phases to facilitate  the actual simulation. In phase 1, the processes agree upon the message alphabet $\stm$ as described in Remark~\ref{rem:messagesAlphabet}. This phase requires $|\stm  | +1$ registers and local states. Recall that the ring must have size bigger than $|\stm|$ for simulating the encoding described in Remark~\ref{rem:messagesAlphabet}. Otherwise, the distributed algorithm will be blocked in this phase. However, our reduction would still work because of two reasons.  First,  our specification will be true for rings smaller than this threshold. This is, in a sense, reducing  the model-checking problem with $\Allrings{}$ prefix to another model-checking problem where the prefix is rephrased to ``All rings of size bigger than $\ell$'' (here, $\ell = |\stm|$). Second, the run which uses only a small amount of tape can be simulated on a big tape. (It maintains the unnecessary cells on the right with the empty tape symbol always. In our simulation these processes will be in the state $(\$,\false)$.) Notice that, the number of processes in the ring is only an upper bound of (rather than exact) space needed by the Turing machine. 

Phase 2 simulates a leader-election protocol, say, the Dolev-Klawe-Rodeh algorithm. The pid of the leader is stored in all processes in a special register $\rleader$. Recall that the leader process will act as the leftmost cell of the tape. A process can always check whether it is the leftmost by comparing the value of $\rleader$ to the register $\idreg$. This check will be used in guards later in transitions involving moving the head of $\tm$ to the left.

Once phase 2 is completed, the configuration of the ring proceeds to represent the initial configuration of $\tm$.  For this, all processes other than the leader will move to the state $(\$,\false)$, i.e., representing the empty tape cell and indicating the absence of the head. The leader process will move to the state $(\$, \true)$.  On taking this transition, the register $\rstate$ of all the processes are set to hold the initial state of the Turing machine.

The simulation of the Turing machine works as follows. Consider a transition of the Turing machine which checks that the current state is $s$ and the current cell contains $a$,  updates the cell content to $b$, moves the head to the left and updates the control state to $s'$. The distributed algorithm will have a transition which moves from local state $(a, \true )$ to $(b,\false)$ which also (i) ensures (by a guard) that $\rstate$ contains the encoding of  $s $, (ii) ensures (by a guard) that it is not the leftmost cell ($\rleader \neq \idreg$), and (iii) sends the encoding of $s'$ to the left. 
For this transition to take place, there are complementary transitions at the receive end which go from $(\text{-}, \false)$ to $(\text{-}, \true)$   upon receiving a value from a neighbor (left or right) to its register $\rstate$. In fact, such a receive transition is enabled for all processes in all the states.
Other transitions of the Turing machine are also implemented similarly. Notice that message transmissions are performed by a process only if $\head = \true$. Notice also that the leader process does not send to left. Also, there are no  forwarding states.

There is actually one subtlety here that arises from the fact that receptions are non-blocking. We have to make sure that a process is aware whether a ``real'' message was received or not. To do so, we introduce a register $r_\bot$, containing a special message $\bot$. Note that the first preliminary phase must indeed be executed for an extended message alphabet that also includes the special symbol $\bot$. For incoming messages, a process will use a special register $r_\mathsf{in}$, which initially contains $\bot$. After executing a receive action, a process will check whether $r_\mathsf{in} \neq r_\bot$, which makes sure that a message has indeed arrived. The subsequent update will then execute $\updcmd{r_\mathsf{in}}{r_\bot}$ to reset $r_\mathsf{in}$.

Finally, the specification $\phi_\tm$ checks that  there is no process  in the state  $(\halt, \true)$. Thus, if the model-checking problem answers negatively, then there is a ring and a run which encodes a valid Turing machine computation on a tape of size bigger than $\stm$ (which also simulates any smaller size tape) and still reaches the \halt state:
\[ \phi_\tm = \Allrings{ \forallpath{{\godown^\ast}}\neg(\halt, \true) }\]
This concludes the proof of Theorem~\ref{thm:undecidable}.
\end{proof}


\section{Proof of Lower bound of Theorem~\ref{thm:main}}
\label{app:LowerboundThmMain}
\newcommand{\Aut}{\mathcal{A}}
\newcommand{\instate}{\textsf{init}}
\newcommand{\finstates}{\textsf{F}}
\newcommand{\Lang}{L}
\newcommand{\encof}{\textsf{EncOf}}
\newcommand{\decof}{\textsf{DecOf}}
\begin{proof} To prove the lower bound, we give a polynomial reduction from the intersection-emptiness problem of finite state automata. That is, given $k$ finite-state automata $\Aut_1, \ldots, \Aut_k$ over a finite alphabet $\Sigma$,  where $\Aut_i = (Q_i, \Delta_i, \instate_i, \finstates_i)$, whether $\bigcap_i \Lang(\Aut_i) = \emptyset$? This problem is known to be PSPACE-complete. 

We will need only unidirectional rings for our reduction.
We construct the distributed algorithm $\DA$ as follows. 

The number of processes in the ring  corresponds to  the length of a candidate word accepted by all the automata $\Aut_i$. Each process thus corresponds to a position in the word. The local state of the process remembers the letter from $\Sigma$ at the respective position. The message contents will be the states of the automata. A preliminary phase sets the message alphabet as per Remark~\ref{rem:messagesAlphabet}. At round $i$ after the preliminary phase, all the processes try to simulate a transition of automaton $\Aut_i$ on the respective position. 
We give the details below. 

In a preliminary phase, the distributed algorithm establishes the finite message alphabet $B = \bigcup_i Q_i$. This requires $|B| +1$ states, registers, and rounds. In case the ring is smaller than $|B|$, the distributed algorithm will be blocked in this phase. However, our reduction would still work because of two reasons.  First,  our specification will be true for rings smaller than this threshold. Second, if a word is accepted by all the automata $\Aut_i$, then acceptance of that word can be simulated on arbitrarily large rings. This will become clear below when we give the actual construction.


The register used for sending the value of a state $s$ to the right is denoted $\encof(s)$. On receiving a value from the left, let $\decof(s)$ be the register against which it is compared to ensure that the received value corresponds to state $s$.

After the preliminary phase, a process non-deterministically moves to a local state from the set $(\Sigma \cup \{\$\}) \times \set{1}$. The special symbol $\$$ marks that a candidate word may start at the right of this process and end at the left of this process. The local state may also remember an index $i$ from $\set{k}$, indicating that it is currently simulating $\Aut_i$. For each $a \in \Sigma$ and $i \le k$, we have a transition of the form
$$\langle (a,i)\textup{:}~\sendright{\encof(s')} \nextcmd \recleft{r} \nextcmd  r = {\decof(s)} \nextcmd \gotocmd{(a,i+1)}\rangle$$
if $(s,a,s') \in \Delta_i$. Further we have 
$$\langle (\$,i)\textup{:}~\sendright{\encof(\instate)} \nextcmd \recleft{r} \nextcmd  r = {\decof(f)} \nextcmd \gotocmd{(\$,i+1)}\rangle$$
if $f \in \finstates_i$.
Notice that the symbol associated to a process does not change in any of these transitions. 

Thus, the number of rounds needed by the distributed algorithm is $\bound = |B| +m + 1$, which is polynomial in the size of the input to intersection emptiness problem of finite state automata. The size of the distributed algorithm $\DA$ is also polynomial. 

Finally,  the $\DataPDLm(\DA)$ formula states that  a state of the form $(\$, k+1)$  cannot be reached:
\[ \phi_m = \Allrings{ \forallpath{{\godown^\ast}}\neg(\$, k+1) } \]

Notice that, if the bounded model checking answers no, then there are a ring, a run, and a marked process $m$ such that $m$ eventually reaches the state $(\$, k+1)$. This means that, on all states $(\$, i)$, $m$ has received a state $f_i \in \finstates_i$. Let $m'$ be 
 the first process on the left of $m$ which has a state of the form $(\$,i)$. Note that $m'$ can be same as $m$.  The word represented by the states of the processes between $m'$ and $m$ is in $\bigcap_i \Lang(\Aut_i)$. Note that, even if this is the empty word (that is, $m'$ is the left neighbor of $m$), it must be in the intersection since $\instate_i \in \finstates_i$ for every automaton $\Aut_i$. On the other hand, if the intersection is non-empty, there is a run that violates the specification.

 Thus, the bounded model checking of $\DA$ answers yes if, and only if, the intersection of the $L(\Aut_i)$ is empty. 
 
 \medskip
 
 This proves the PSPACE lower bound stated in Theorem~\ref{thm:main}.
 \end{proof}

 
\section{Proof of Theorem~\ref{thm:icpdl}}

We can restrict to pictures of height $k=\bound$ (rather than $k \le \bound$), since checking satisfiability for every height separately does not change the complexity. We reduce the problem to words, for which \ICPDL satisfiability is known to be PSPACE-complete \cite{Goeller2009} (since formulas from \ICPDL have bounded intersection width). A picture $\pict=(n,k,\lambda)$ is considered as the word $\pict[1] \cdot \ldots \cdot \pict[n] \in \Trans^+$.
Thus, the columns are written horizontally rather than vertically. When translating an \ICPDL formula over tables into an \ICPDL formula over words, going to the left or right involves some modulo counting: $\goleft$ is translated to $\goleft^{k+1}$, and $\goright$ is translated to $\goright^{k+1}$. An additional difficulty stems from the fact that we allow automata as path expressions, but it is straightforward to integrate them into the construction of an alternating two-way automaton from \cite{Goeller2009}.


 
\newcommand{\TDA}{\mathcal{T}_{\DA}}

\section{Proof of Lemma~\ref{lem:datapdl}}

Let us first introduce some notation.
Let $\TDA = \{\pictofrun{\run} \mid \run$ is a run of $\DA\}$.
For a table $T \in \Pictures$, let $\Runs{\pict}=\{\run \mid \run$ is run of $\DA$ such that $T_\run = T\}$.

A \emph{pseudo} ($\ring$-)\emph{run} of $\DA$ is like an ($\ring$-)run $\smash{\run = {\conf_0 \confrel{t^1} \conf_1 \confrel{t^2} \ldots \confrel{t^k} \conf_k}}$, but conditions 1.-3.\ are not checked. That is, target and source states are not necessarily matching, and $=$- and $<$-guards are ignored. Thus, every run is a pseudo run, but not vice versa.
We define $T_\run$ and $\Coord{\run}$ in exactly the same way as for runs.

Given a (pseudo) run ${\run = {\conf_0 \confrel{t^1} \conf_1 \confrel{t^2} \ldots \confrel{t^k} \conf_k}}$ of $\DA$ (where $\conf_j = (s_1^j,\ldots,s_n^j,\regmap_1^j,\ldots,\regmap_n^j)$) and $(i,j) \in \Coord{\run}$, we set $\run_i^j = \regmap_i^j$ (abusing notation). Moreover, for $j \ge 1$, $\hat{\run}_i^j = \hat{\regmap}_i^j$ defines the corresponding $j$-th intermediate register assignment, which was defined in Section~\ref{sec:algorithms} to obtain the mapping $\regmap_i^j$. Finally, we set $\hat{\run}_i^0 = \run_i^0$.

\medskip

To prove Lemma~\ref{lem:datapdl}, we will need two further lemmas:

\begin{lemma}\label{lem:pathaut}
For all pseudo runs $\run$ of $\DA$, coordinates $(i,j),(i',j') \in \Coord{\run}$, and registers $r \in \Reg$, the following hold:
\begin{itemize}\itemsep=0.5ex
\item[(a)] $((i,j),(i',j')) \in \Sem{\pathaut_{r}^{1}}{T_\run} ~\Longleftrightarrow~ \bigl(\run_{i}^{0}(\idreg) = \hat{\run}_{i'}^{j'}(r) ~\wedge~ j=0\bigr)$
\item[(b)] $((i,j),(i',j')) \in \Sem{\pathaut_{r}^{2}}{T_\run} ~\Longleftrightarrow~ \bigl(\run_{i}^{0}(\idreg) = \run_{i'}^{j'}(r) ~\wedge~ j=0\bigr)$
\end{itemize}
\end{lemma}

\begin{proof}
Let the pseudo $\ring$-run in question be given by
$\smash{\run = {\conf_0 \confrel{t^1} \conf_1 \confrel{t^2} \ldots \confrel{t^k} \conf_k}}$
where $t^j = (t_1^j,\ldots,t_n^j) \in \Trans^n$.

To be able to perform an induction, we show a more general statement that captures both (a) and (b). To this aim, we define the automaton $\pathaut_{r}^{0}$ in the expected manner, i.e., where the only final state is $(0,r)$.
We will show, for all $\ph \in \{0,1,2\}$,
\begin{align}
((i,j),(i',j')) \in \Sem{\pathaut_{r}^{\ph}}{T_\run} ~\Longleftrightarrow~ \bigl(\run_{i}^{0}(\idreg) = \run_{i'}^{j'}[\ph](r) ~\wedge~ j=0\bigr)\,.
\end{align}
Here, $\run_{i'}^{j'}[2]$ refers to $\run_{i'}^{j'}$ and $\run_{i'}^{j'}[1]$ refers to $\hat{\run}_{i'}^{j'}$ (recall that $\hat{\run}_{i'}^0 = \run_{i'}^0$). For $j' \ge 1$, we let
$\run_{i'}^{j'}[0](r)$ refer to the value of $r$ at position $(i',j')$ before reception.
Finally, we set $\run_{i'}^{0}[0] = \run_{i'}^{0}[1] ( = \run_{i'}^0)$.

\newcommand{\Confs}{\mathit{Conf}}

Before we come to the actual proof of (1), we define the relation
\[{\longrightarrow}_\run \subseteq \Confs \times \{\mathit{loc},\mathit{upd},\mathit{next},\mathit{msg}\} \times \Confs\]
where $\Confs = \Coord{\run} \times \{0,1,2\} \times \Reg$. The idea is that $\longrightarrow_\run$ captures the flow of pids in $\run$. We let $\longrightarrow_\run$ be the least relation satisfying the following:
\begin{itemize}
\item $(i,j,r,0) \xrightarrow{\mathit{loc}}_\run (i,j,r,1)$ if there are no $r',i'$ such that $\auxright{r'}{i'}{r}{i}$ (in step $\conf_{j-1} \confrel{t^j} \conf_j$)
\item $(i,j,r,1) \xrightarrow{\mathit{upd}}_\run (i,j,r',2)$ if $r \neq r'$ and $(\updcmd{r'}{r}) \in t_i^j$, or $r=r'$ and $(\updcmd{r}{r''}) \not\in t_i^j$ for all $r'' \neq r$
\item $(i,j,r,2) \xrightarrow{\mathit{next}}_\run (i,j+1,r,0)$
\item $(i,j,r,0) \xrightarrow{\mathit{msg}}_\run (i',j,r',1)$ if $\auxright{r}{i}{r'}{i'}$ or $\auxleft{r}{i}{r'}{i'}$ (in step $\conf_{j-1} \confrel{t^j} \conf_j$)
\end{itemize}
Note that $(i,j,r,\ph) \xrightarrow{\theta}_\run (i',j',r',\ph')$ immediately implies $\run_{i}^{j}[\ph](r)=\run_{i'}^{j'}[\ph'](r')$.
We will show that, moreover, we have
\begin{align}
(i,j,r,\ph) \xrightarrow{\theta}_\run (i',j',r',\ph') ~\Longleftrightarrow~ ((i,j),(i',j')) \in \Sem{\theta_{r,r'}^{\ph,\ph'}}{T_\run}
\end{align}
To prove this, we distinguish four cases:
\begin{itemize}\itemsep=1.5ex
\item Suppose $\theta = \mathit{loc}$. Then, we can assume $\ph = 0$ and $\ph' = 1$. We have
\[
\begin{array}{rl}
& (i,j,r,0) \xrightarrow{\theta}_\run (i',j',r',1)\\[0ex]
\Longleftrightarrow & r=r' ~\wedge~ (i,j)=(i',j') ~\wedge~
\text{$\neg\exists \bar{r},\bar{i}$ such that $\auxright{\bar{r}}{\bar{i}}{r}{i}$ (in step $\conf_{j-1} \confrel{t^j} \conf_j$)}\\[0.5ex]
\Longleftrightarrow &  r=r' ~\wedge~ ((i,j),(i',j')) \in \Sem{\test{\bigwedge_{\bar{r} \in \Reg}\neg\existspath{(\msgform{\bar{r}}{r})^{-1}}}}{T_\run}\\[0.5ex]
\Longleftrightarrow & ((i,j),(i',j')) \in \Sem{\mathit{loc}_{r,r'}^{0,1}}{T_\run}
\end{array}
\]

\item Suppose $\theta = \mathit{upd}$. We can assume $\ph = 1$ and $\ph' = 2$.
We distinguish two subcases.
\begin{enumerate}\itemsep=1ex
\item Suppose $r \neq r'$. Then, we have
\[
\begin{array}{rl}
& (i,j,r,1) \xrightarrow{\theta}_\run (i',j',r',2)\\[0.5ex]
\Longleftrightarrow & (i,j) = (i',j') ~\wedge~ (\updcmd{r'}{r}) \in t_{i'}^{j'}\\[0.5ex]
\Longleftrightarrow & ((i,j),(i',j')) \in \Sem{\test{\updcmd{r'}{r}}}{T_\run}\\[0.7ex]
\Longleftrightarrow & ((i,j),(i',j')) \in \Sem{\mathit{update}_{r,r'}^{1,2}}{T_\run}
\end{array}
\]

\item Suppose $r = r'$. Then,
\[
\begin{array}{rl}
& (i,j,r,1) \xrightarrow{\theta}_\run (i',j',r',2)\\[0.5ex]
\Longleftrightarrow & (i,j) = (i',j') ~\wedge~ (\updcmd{r}{\bar{r}}) \not\in t_{i'}^{j'} \text{ for all } \bar{r} \neq r\\[0.5ex]
\Longleftrightarrow & ((i,j),(i',j')) \in \Sem{\test{ \bigwedge_{\bar{r} \neq r} \neg(\updcmd{r}{\bar{r}})}}{T_\run}\\[0.7ex]
\Longleftrightarrow & ((i,j),(i',j')) \in \Sem{\mathit{update}_{r,r'}^{1,2}}{T_\run}
\end{array}
\]
\end{enumerate}

\item Suppose $\theta = \mathit{next}$. We can assume $\ph = 2$ and $\ph' = 0$. We have
\[
\begin{array}{rl}
& (i,j,r,2) \xrightarrow{\theta}_\run (i',j',r',0)\\[0.8ex]
\Longleftrightarrow & r = r' ~\wedge~ i=i' ~\wedge~ j' = j+1\\[0.5ex]
\Longleftrightarrow & r=r' ~\wedge~ ((i,j),(i',j')) \in \Sem{\godown}{T_\run}\\[0.7ex]
\Longleftrightarrow & ((i,j),(i',j')) \in \Sem{\mathit{next}_{r,r'}^{2,1}}{T_\run}
\end{array}
\]
\end{itemize}

We are now ready to prove (1).

\medskip
\noindent
($\Rightarrow$):
First note that $((i,j),(i',j')) \in \Sem{\pathaut_{r}^{\ph}}{T_\run}$ always implies $j=0$, since the automaton has to read $\test{\neg{\existspath{\goup}}}$ before it can accept at all (its initial state $\iota$ is not a final state).

Consider an (accepting) execution
\[\iota \xrightarrow{\pi_1} (r_1,\ph_1) \xrightarrow{\pi_2} \ldots \xrightarrow{\pi_\ell} (r_\ell,h_\ell)=(\ph,r)\]
of $\pathaut_r^h$, with $\ell \ge 1$, $\pi_1 = \test{\neg\existspath{\goup}}$, and $\pi_l = (\theta_l)_{r_{l-1},r_l}^{\ph_{l-1},\ph_l}$ for all $l \in \{2,\ldots,\ell\}$, connecting $(u,0)$ with $(i,j)$.
That is, $((u,0),(i,j)) \in \Sem{\pi_1 \cdot \ldots \cdot \pi_\ell}{T_\run}$. We have to show $\run_{u}^{0}(\idreg) = \run_{i}^{j}[\ph](r)$.

There are positions $(u,0)=(i_0,j_0),(i_1,j_1),\ldots,(i_\ell,j_\ell)=(i,j) \in \Coord{\run}$ such that $((i_{l-1},j_{l-1}),(i_{l},j_{l})) \in \Sem{\pi_l}{T_\run}$ for all $l \in \set{\ell}$. By (2), we obtain
\[(i_1,j_1,r_1,\ph_1) \xrightarrow{\theta_2}_\run (i_2,j_2,r_2,\ph_2) \xrightarrow{\theta_3}_\run \ldots \xrightarrow{\theta_\ell}_\run (i_\ell,j_\ell,r_\ell,\ph_\ell)\,.\]
This implies $\run_{i_1}^{j_1}[\ph_1](r_1) = \run_{i_\ell}^{j_\ell}[\ph_\ell](r_\ell)$, which equals $\run_{i}^{j}[\ph](r)$.
Since $\pi_1 = \test{\neg\existspath{\goup}}$, we also have $(u,0)=(i_1,j_1)$ and, therefore, $\run_{u}^{0}(\idreg) = \run_{i_1}^{j_1}[\ph_1](r_1)$. We conclude $\run_{u}^{0}(\idreg) = \run_{i}^{j}[\ph](r)$.

\medskip
\noindent
($\Leftarrow$):
Suppose $\run_{u}^{0}(\idreg) = \run_{i}^{j}[\ph](r)$. We will show that $((u,0),(i,j)) \in \Sem{\pathaut_{r}^{\ph}}{T_\run}$.

By the semantics of $\DA$, pid $\run_{u}^{0}(\idreg)$ has to be transmitted along transitions or messages. Thus, there are $\ell \ge 1$, positions $(i_1,j_1),\ldots,(i_\ell,j_\ell)=(i,j) \in \Coord{\run}$, registers $r_1,\ldots,r_\ell=r$, stages $0=\ph_1,\ldots,\ph_\ell=\ph \in \{0,1,2\}$, and
$\theta_2,\ldots,\theta_\ell \in \{\mathit{loc},\mathit{upd},\mathit{next},\mathit{msg}\}$
such that
\begin{itemize}
\item $((u,0),(i_1,j_1)) \in \Sem{\neg\existspath{\goup}}{T_\run}$ (therefore, $(u,0)=(i_1,j_1)$), and

\item $(i_{l-1},j_{l-1},r_{l-1},\ph_{l-1}) \xrightarrow{\theta_\ell}_\run (i_{l},j_{l},r_{l},\ph_{l})$
for all $l \in \{2,\ldots,\ell\}$.
\end{itemize}
By (2), we have
\[((i_{l-1},j_{l-1}),(i_l,j_l)) \in \Sem{(\theta_l)_{r_{l-1},r_l}^{\ph_{l-1},\ph_l}}{T_\run}\]
for all $l \in \{2,\ldots,\ell\}$.
We deduce
\[((u,0),(i,j)) = ((u,0),(i_l,j_l)) \in \Sem{\pathaut_{r_\ell}^{\ph_\ell}}{T_\run} = \Sem{\pathaut_{r}^{\ph}}{T_\run}\,.\]
This concludes the proof of Lemma~\ref{lem:pathaut}.
\end{proof}

\begin{lemma}\label{lem:piless}
For all $T=(n,k,\lambda) \in \TDA$ and $i,i' \in \set{n}$, we have
\[
((i,0),(i',0)) \in \Sem{\pi_<}{T}
~\Longleftrightarrow~ \forall \run \in \Runs{T}: \run_{i}^{0}(\idreg) < \run_{i'}^{0}(\idreg)\,.
\]
\end{lemma}

\newcommand{\bi}{u}
\newcommand{\bip}{u'}
\begin{proof}
There are two directions to show.

\medskip
\noindent
($\Rightarrow$): Suppose $((i,0),(i',0)) \in \Sem{\pi_<}{T}$. Then, there are $\ell \ge 1$ and $i=i_0,\ldots,i_\ell=i'$ such that
\[{((i_{l-1},0),(i_{l},0)) \in \Sem{\sum_{r,r' \in \Reg} \pathaut_{r}^{1} \cdot \test{r < r'} \cdot (\pathaut_{r'}^{1})^{-1}~}{T}}\]
for all $l \in \set{\ell}$.
Let $\run \in \Runs{T}$. By Lemma~\ref{lem:pathaut}, we have $\smash{\run_{i_{l-1}}^{0}(\idreg) < \run_{i_{l}}^{0}(\idreg)}$ for all $l \in \set{\ell}$. We deduce $\smash{\run_{i}^{0}(\idreg) < \run_{i'}^{0}(\idreg)}$.

\medskip
\noindent
($\Leftarrow$): We denote the processes in question by $\bi$ and $\bip$. Suppose that $((\bi,0),(\bip,0)) \not\in \Sem{\pi_<}{T}$. We are going to show that there is $\run \in \Runs{T}$ such that $\run_{\bi}^{0}(\idreg) \geq \run_{\bip}^{0}(\idreg)$. Let ${\prec} = \{(i,i') \mid ((i,0),(i',0)) \in \Sem{\pi_<}{T}\}$. In particular, $\bi \not\prec \bip$. By direction ($\Rightarrow$), we have that $\prec$ is a (strict) partial order.

Let $\ring=(n:p_1,\ldots,p_n)$ be any ring such that (i) $p_{\bi} \ge p_{\bip}$ and (ii) for all $i,i' \in \set{n}$, $i \prec i'$ implies $p_i < p_{i'}$. Since $\prec$ is a strict partial order and $\bi \not\prec \bip$, such a ring must exist.
Now, note that there is a unique \emph{pseudo} $\ring$-run
\[\smash{\run = {\conf_0 \confrel{t^1} \conf_1 \confrel{t^2} \ldots \confrel{t^k} \conf_k}}\]
(where $t^j = (t_1^j,\ldots,t_n^j) \in \Trans^n$) such that $T_\run=T$.
We will show that $\run$ is indeed also an $\ring$-run, which concludes the proof.

Let $(i,j) \in \Coord{T}$ and $r,r' \in \Reg$ such that $(r < r') \in t_{i}^{j}$. We have to show that $\hat{\run}_i^j(r) < \hat{\run}_i^j(r')$.
By Lemma~\ref{lem:pathaut}, there are $o,o' \in \set{n}$ such that
\begin{itemize}\itemsep=1ex
\item $\run_o^0(\idreg)=\hat{\run}_i^j(r)$ and $\run_{o'}^0(\idreg)=\hat{\run}_i^j(r')$, and
\item $((o,0),(i,j)) \in \Sem{\pathaut_{r}^{1}}{T_\run} \text{ and } ((o',0),(i,j)) \in \Sem{\pathaut_{r'}^{1}}{T_\run}$.
\end{itemize}
The latter implies
\[((o,0),(o',0)) \in \Sem{\pathaut_{r}^{1} \cdot \test{r < r'} \cdot (\pathaut_{r'}^{1})^{-1}}{T_\run}\,.\]
In particular, $((o,0),(o',0)) \in \Sem{\pi_<}{T_\run}$. We deduce $o \prec o'$. This implies $\run_o^0(\idreg) < \run_{o'}^0(\idreg)$. We conclude that $\run_i^j(r) < \run_i^j(r')$.

Finally, let $(i,j) \in \Coord{T}$ and $r,r' \in \Reg$ such that $(r = r') \in t_{i}^{j}$. Since $\Runs{T} \neq \emptyset$, there is a run that validates guard $r=r'$ at coordinate $(i,j)$. By Lemma~\ref{lem:pathaut}, this is actually true for all pseudo runs of $T$. We deduce $\run_i^j(r) = \run_i^j(r')$.

Note that run condition 1.\ is satisfied, since $T \in \TDA$. This concludes the proof.
\end{proof}

We will now proceed to the proof of Lemma~\ref{lem:datapdl}.

\begin{proof}[Proof of Lemma~\ref{lem:datapdl}.]
Recall that we have to show $L(\lcpdl_\DA) = \TDA$, where $\lcpdl_\DA \formdef \lcpdl_= \wedge \lcpdl_< \wedge \lcpdl_{\mathsf{col}}$. 

\medskip
\noindent
($\subseteq$):
Let $T=(n,k,\lambda) \in L(\psi_\DA)$. We will show $T \in \TDA$ by constructing a run
$\run$ of $\DA$ such that $T_\run=T$.

Again, let ${\prec} = \{(i,i') \mid ((i,0),(i',0)) \in \Sem{\pi_<}{T}\}$.
As $T,(1,0) \models \lcpdl_< \formdef \neg\existspath{\goright^\ast} \loopform{\pi_<}$, we have that $\prec$ is a strict partial order. Choose any ring $\ring=(n:p_1,\ldots,p_n)$ such that,
for all $i,i' \in \set{n}$, $i \prec i'$ implies $p_i < p_{i'}$. There is a unique pseudo $\ring$-run
\[{\run = {\conf_0 \confrel{t^1} \conf_1 \confrel{t^2} \ldots \confrel{t^k} \conf_k}}\]
of $\DA$ such that $T_\run=T$.
Let $j \in \set{k}$. We have to show $\conf_{j-1} \confrel{t^j} \conf_j$ where, this time, all run conditions are checked. Condition 4.\ of the definition of $\rightsquigarrow$ is satisfied thanks to the definition of a pseudo run. Condition 1.\ is ensured by $T \in L(\psi_\mathsf{col})$. Let $i \in \set{n}$ and suppose $(r=r') \in t_i^j$. We have $T,(i,j) \models \loopform{(\pathaut_{r}^{1})^{-1} \cdot \pathaut_{r'}^{1}}$. By Lemma~\ref{lem:pathaut}, we have $\hat{\run}_{i}^{j}(r) = \hat{\run}_{i}^{j}(r')$.
Finally, suppose $(r<r') \in t_i^j$. We proceed like in the reverse direction of the proof of Lemma~\ref{lem:piless} to show that $\hat{\run}_i^j(r) < \hat{\run}_i^j(r')$.

Altogether, it follows that $\run$ is a run.

\medskip
\noindent
($\supseteq$): Let $T=(n,k,\lambda) \in \Pictures$ such that $T \not\in L(\psi_\DA)$. To show $T \not\in \TDA$, we distinguish three (non-disjoint) cases.
\begin{itemize}\itemsep=1ex
\item Suppose $T \not\in L(\lcpdl_{\mathsf{col}})$. Obviously, this implies $T \not\in \TDA$.

\item Suppose $T \not\in (\lcpdl_=)$. Recall that
\[\lcpdl_= ~\formdef~ \forallpath{(\goright + \godown)^\ast} \bigwedge_{r,r' \in \Reg} \Bigl(r = r' ~\Rightarrow \loopform{(\pathaut_{r}^{1})^{-1} \cdot \pathaut_{r'}^{1}}\Bigr)\]
Thus, there are a coordinate $(i,j) \in \set{n} \times \setz{k}$ and registers $r_1,r_2 \in \Reg$ such that we have $(r_1 = r_2) \in T[i,j]$ and $T,(i,j) \not\models \loopform{(\pathaut_{r_1}^{1})^{-1} \cdot \pathaut_{r_2}^{1}}$. Towards a contradiction, suppose there is $\run \in \Runs{T}$. By Lemma~\ref{lem:pathaut}, there are (unique) $i_1,i_2 \in \set{n}$ such that $\run_{i_1}^{0}(\idreg) = \hat{\run}_{i}^{j}(r_1)$ and $\run_{i_2}^{0}(\idreg) = \hat{\run}_{i}^{j}(r_2)$, as well as $((i_1,0),(i,j)) \in \Sem{\pathaut_{r_1}^{1}}{T}$ and $(i_2,0),(i,j)) \in \Sem{\pathaut_{r_2}^{1}}{T}$. Since $T,(i,j) \not\models \loopform{(\pathaut_{r_1}^{1})^{-1} \cdot \pathaut_{r_2}^{1}}$, we have that $i_1 \neq i_2$. We deduce $\hat{\run}_{i}^{j}(r_1) \neq \hat{\run}_{i}^{j}(r_2)$, which contradicts $(r_1 = r_2) \in T[i,j]$. Altogether, we obtain $T \not\in \TDA$.

\item Suppose $T \not\in L(\lcpdl_<)$ where $\lcpdl_< ~\formdef~ \neg\existspath{\goright^\ast} \loopform{\pi_<}$.
Then, there is $i \in \set{n}$ such that $T,(i,0) \models \loopform{\pi_<}$. By Lemma~\ref{lem:piless}, we have $\run_{i}^{0}(\idreg) < \run_{i}^{0}(\idreg)$ for all runs $\run \in \Runs{T}$. Thus, $\Runs{T} = \emptyset$ and, therefore, $T \not\in \TDA$.
\end{itemize}
This concludes the proof of Lemma~\ref{lem:datapdl}.
\end{proof}



\section{Proof of Lemma~\ref{lem:lcpdl}}

We show a more general statement. First, call a local \DataPDLm formula $\phi$ \emph{good} if it does not contain any guard of the form $<$ or $\le$. Recall that we set $\TDA = \{\pictofrun{\run} \mid \run$ is a run of $\DA\}$ and, for a table $T \in \Pictures$, $\Runs{\pict}=\{\run \mid \run$ is run of $\DA$ such that $T_\run = T\}$.

\medskip

We will simultaneously show the following statements:
\begin{itemize}
\item For all local $\DataPDLm(\DA)$ formulas $\phi$:
\begin{itemize}
\item[(a)] for all $\pict \in \TDA$ and $(i,j) \in \Coord{T}$,
\[T,(i,j) \models \transl{\phi}
~~\Longleftrightarrow~~
\run,1,(i,j) \models \phi \text{ for all } \run \in \Runs{T}.\]
\end{itemize}

\item For all good local $\DataPDLm(\DA)$ formulas $\phi$:
\begin{itemize}
\item[(b)] for all runs $\run$ of $\DA$ and all $(i,j) \in \Coord{\run}$,
\[T_\run,(i,j) \models \transl{\phi}
~~\Longleftrightarrow~~
\run,1,(i,j) \models \phi.\]
\end{itemize}

\item For all $\DataPDLm(\DA)$ path formulas $\pi$:
\begin{itemize}
\item[(c)] for all runs $\run$ of $\DA$, we have $\Sem{\tilde{\pi}}{T_\run} = \Sem{\pi}{\run,1}$.
\end{itemize}
\end{itemize}

\noindent
We first consider local formulas. We proceed by induction on the structure of $\phi$. Note that (b) is a stronger statement: when we show that (b) holds for a formula, then (a) holds for that formula, too.

\begin{itemize}\itemsep=3ex
\item Suppose $\phi = \marked$. It is enough to show (b). Recall that $\transl{\marked} = \neg\existspath{\goleft}$. We have $T_\run,(i,j) \models \neg\existspath{\goleft} ~\Longleftrightarrow~ i=1 ~\Longleftrightarrow~ \run,1,(i,j) \models \marked$.

\item Suppose $\phi = s \in \States$. Again, it is enough to show (b). Recall that $\transl{s} = \gotocmd{s}$. By the definition of runs, the semantics of \DataPDLm, and $T_\run$, we have that $T_\run,(i,j) \models \gotocmd{s} ~\Longleftrightarrow~ \run,1,(i,j) \models s$.

\item Consider $\neg\phi$. Then, $\phi$ is a good formula. Recall that $\transl{\neg\phi} = \neg\transl{\phi}$. We have $T_\run,(i,j) \models \neg\transl{\phi} ~\Longleftrightarrow~ T_\run,(i,j) \not\models \transl{\phi} ~\Longleftrightarrow~ \text{(by I.H.(b)) } \run,1,(i,j) \not\models \phi ~\Longleftrightarrow~ \run,1,(i,j) \models \neg\phi$.

\item Suppose $\phi = (\phi_1 \wedge \phi_2)$.
\begin{itemize}\itemsep=1ex
\item[(a)] We have
\[\begin{array}{cl}
& T,(i,j) \models \transl{\phi_1} \wedge \transl{\phi_2}\\[1ex]
\Longleftrightarrow & T,(i,j) \models \transl{\phi_1} \text{ and } T,(i,j) \models \transl{\phi_2}\\[1ex]
\smash{\stackrel{\text{I.H.(a)}}{\Longleftrightarrow}} &
\bigl(\run,1,(i,j) \models \phi_1 \text{ for all } \run \in \Runs{T}\bigr) \text{ and }\\[0.5ex]
& \bigl(\run,1,(i,j) \models \phi_2 \text{ for all } \run \in \Runs{T}\bigr)\\[1ex]
\Longleftrightarrow & \run,1,(i,j) \models \phi_1 \wedge \phi_2 \text{ for all } \run \in \Runs{T}
\end{array}\]

\item[(b)] Suppose $\phi_1$ and $\phi_2$ are good. We have
\[\begin{array}{cl}
& T_\run,(i,j) \models \transl{\phi_1} \wedge \transl{\phi_2}\\[1ex]
\Longleftrightarrow & T_\run,(i,j) \models \transl{\phi_1} \text{ and } T_\run,(i,j) \models \transl{\phi_2}\\[1ex]
\smash{\stackrel{\text{I.H.(b)}}{\Longleftrightarrow}} & \run,1,(i,j) \models \phi_1 \text{ and } \run,1,(i,j) \models \phi_2\\[1ex]
\Longleftrightarrow & \run,1,(i,j) \models \phi_1 \wedge \phi_2
\end{array}\]
\end{itemize}

\item Consider $\phi = (\phi_1 \Rightarrow \phi_2)$. Then, $\phi_1$ is good.
\begin{itemize}\itemsep=1ex
\item[(a)] There are two directions to show:
\begin{itemize}
\item[($\Rightarrow$):]
We have 
\[\begin{array}{cl}
& T,(i,j) \models \transl{\phi_1} \Rightarrow \transl{\phi_2}\\[1ex]
\Longrightarrow & T,(i,j) \not\models \transl{\phi_1} \text{ or } T,(i,j) \models \transl{\phi_2}\\[1ex]
\smash{\stackrel{\text{I.H.(b),(a)}}{\Longrightarrow}} &
\bigl(\run,1,(i,j) \not\models \phi_1 \text{ for all } \run \in \Runs{T}\bigr) \text{ or }\\[0.5ex]
& \bigl(\run,1,(i,j) \models \phi_2 \text{ for all } \run \in \Runs{T}\bigr)\\[1ex]
\Longrightarrow & \run,1,(i,j) \models \phi_1 \Rightarrow \phi_2 \text{ for all } \run \in \Runs{T}
\end{array}\]

\item[($\Leftarrow$):] We have 
\[\begin{array}{cl}
& T,(i,j) \not\models \transl{\phi_1} \Rightarrow \transl{\phi_2}\\[1ex]
\Longrightarrow & T,(i,j) \models \transl{\phi_1} \text{ and } T,(i,j) \not\models \transl{\phi_2}\\[1ex]
\smash{\stackrel{\text{I.H.(b),(a)}}{\Longrightarrow}} &
\bigl(\run,1,(i,j) \models \phi_1 \text{ for all } \run \in \Runs{T}\bigr) \text{ and }\\[0.5ex]
& \bigl(\run,1,(i,j) \not\models \phi_2 \text{ for some } \run \in \Runs{T}\bigr)\\[1ex]
\Longrightarrow & \run,1,(i,j) \not\models \phi_1 \Rightarrow \phi_2 \text{ for some } \run \in \Runs{T}
\end{array}\]
\end{itemize}

\item[(b)] Here, we require that both $\phi_1$ and $\phi_2$ are good. Then,
\[\begin{array}{cl}
& T_\run,(i,j) \models \transl{\phi_1} \Rightarrow \transl{\phi_2}\\[1ex]
\Longleftrightarrow & T_\run,(i,j) \not\models \transl{\phi_1} \text{ or } T_\run,(i,j) \models \transl{\phi_2}\\[1ex]
\smash{\stackrel{\text{I.H.(b)}}{\Longleftrightarrow}} & \run,1,(i,j) \not\models \phi_1 \text{ or } \run,1,(i,j) \models \phi_2\\[1ex]
\Longleftrightarrow & \run,1,(i,j) \models \phi_1 \Rightarrow \phi_2
\end{array}\]
\end{itemize}

\item Consider formula $\forallpath{\pi}\phi$. Let $x=(i,j)$. For a set $A \subseteq \Coord{T} \times \Coord{T}$, let $A(x) = \{x' \in \Coord{T} \mid (x,x') \in A\}$.
\begin{itemize}\itemsep=1ex
\item[(a)] We have
\[\begin{array}{cl}
& T,x \models \forallpath{\transl{\pi}}\transl{\phi}\\[1ex]

\Longleftrightarrow & \forall x' \in \Sem{\transl\pi}{T}(x):  T,x' \models \transl{\phi}\\[1ex]

\smash{\stackrel{\text{I.H.(a)}}{\Longleftrightarrow}} &
\forall x' \in \Sem{\transl\pi}{T}(x):  \forall \run \in \Runs{T}: \run,1,x' \models \phi\\[1ex]

\Longleftrightarrow & \forall \run \in \Runs{T}: \forall x' \in \Sem{\transl\pi}{T_\chi}(x): \run,1,x' \models \phi\\[1ex]

\smash{\stackrel{\text{I.H.(c)}}{\Longleftrightarrow}} &
\forall \run \in \Runs{T}: \forall x' \in \Sem{\pi}{\chi,1}(x): \run,1,x' \models \phi\\[1ex]

\Longleftrightarrow &
\forall \run \in \Runs{T}: \run,1,x \models \forallpath{\pi}\phi\\[1ex]
\end{array}\]

\item[(b)] Suppose $\phi$ is good. We have
\[\begin{array}{cl}
& T_\run,x \models \forallpath{\transl\pi}\transl\phi\\[1ex]

\Longleftrightarrow & \forall x' \in \Sem{\transl\pi}{T_\run}(x):  T_\run,x' \models \transl{\phi}\\[1ex]

\smash{\stackrel{\text{I.H.(b)}}{\Longleftrightarrow}} &
\forall x' \in \Sem{\transl\pi}{T_\run}(x):  \run,1,x' \models \phi\\[1ex]

\smash{\stackrel{\text{I.H.(c)}}{\Longleftrightarrow}} &
\forall x' \in \Sem{\pi}{\run,1}(x):  \run,1,x' \models \phi\\[1ex]

\Longleftrightarrow &
 \run,1,x \models \forallpath{\pi}\phi
\end{array}\]
\end{itemize}


\item Suppose $\phi = \existsp{r_1}{r_2}{\pi_1}{\pi_2}{\le}$. Then, $\pi_1$ and $\pi_2$ are both 
unambiguous. By I.H.(c), $\transl{\pi_1}$ and $\transl{\pi_2}$ are unambiguous (wrt.\ symbolic runs).
We show (a):
\begin{center}
\begin{tabular}{rl}
& $T,(i,j) \models \transl{\existsp{r_1}{r_2}{\pi_1}{\pi_2}{\le}}$\\[2ex]

$\Longleftrightarrow$ &
\parbox[t]{35em}{$T,(i,j) \models \loopform{\transl{\pi_1} \cdot (\pathaut_{r_1}^{2})^{-1} \cdot (\pi_< + \stay) \cdot \pathaut_{r_2}^{2} \cdot (\transl{\pi_2})^{-1}}$}
\end{tabular}\vspace{1ex}

\begin{tabular}{rl}
$\Longleftrightarrow$ &
\parbox[t]{35em}{
there are coordinates $(i_1,j_1),(i_2,j_2),(i_1',0),(i_2',0) \in \Coord{T}$ such that:
\begin{enumerate}\itemsep=1ex
\item $((i,j),(i_1,j_1)) \in \Sem{\transl{\pi_1}}{T}$ and $((i,j),(i_2,j_2)) \in 
\Sem{\transl{\pi_2}}{T}$

\item $((i_1',0),(i_1,j_1)) \in \Sem{\pathaut_{r_1}^{2}}{T}$ and $((i_2',0),(i_2,j_2)) \in \Sem{\pathaut_{r_2}^{2}}{T}$

\item $((i_1',0),(i_2',0)) \in \Sem{\pi_<}{T}$ or $i_1' = i_2'$
\end{enumerate}}
\end{tabular}\vspace{1ex}

\begin{tabular}{rl}
$\stackrel{(\ast)}{\Longleftrightarrow}$ &
\parbox[t]{35em}{
there exist coordinates $(i_1,j_1),(i_2,j_2),(i_1',0),(i_2',0) \in \Coord{T}$ such that:
\begin{enumerate}\itemsep=1ex
\item $\forall \run \in \Runs{T}: ((i,j),(i_1,j_1)) \in \Sem{\pi_1}{\run,1}$ and $((i,j),(i_2,j_2)) \in 
\Sem{\pi_2}{\run,1}$

\item $\smash{\forall \run \in \Runs{T}: \run_{i_1'}^{0}(\idreg) = \run_{i_1}^{j_1}(r_1)}$ and
$\smash{\run_{i_2'}^{0}(\idreg) = \run_{i_2}^{j_2}(r_2)}$

\item $\smash{\bigl(\forall \run \in \Runs{T}: \run_{i_1'}^{0}(\idreg) < \run_{i_2'}^{0}(\idreg)\bigr)}$ or $i_1' = i_2'$
\end{enumerate}}
\end{tabular}\vspace{1ex}

\begin{tabular}{rl}
$\Longleftrightarrow$ &
\parbox[t]{35em}{
there exist coordinates $(i_1,j_1),(i_2,j_2),(i_1',0),(i_2',0) \in \Coord{T}$ such that:
\begin{enumerate}\itemsep=1ex
\item $\forall \run \in \Runs{T}: ((i,j),(i_1,j_1)) \in \Sem{\pi_1}{\run,1}$ and $((i,j),(i_2,j_2)) \in 
\Sem{\pi_2}{\run,1}$

\item $\smash{\forall \run \in \Runs{T}: \run_{i_1'}^{0}(\idreg) = \run_{i_1}^{j_1}(r_1)}$ and
$\smash{\run_{i_2'}^{0}(\idreg) = \run_{i_2}^{j_2}(r_2)}$

\item $\smash{\forall \run \in \Runs{T}: \run_{i_1'}^{0}(\idreg) < \run_{i_2'}^{0}(\idreg)}$ or
$\smash{\run_{i_1'}^{0}(\idreg) = \run_{i_2'}^{0}(\idreg)}$
\end{enumerate}}
\end{tabular}\vspace{1ex}

\begin{tabular}{rl}
$\stackrel{(\ast\ast)}{\Longleftrightarrow}$ &
\parbox[t]{35em}{
for all $\run \in \Runs{T}$, there are $(i_1,j_1),(i_2,j_2),(i_1',0),(i_2',0) \in \Coord{T}$ such that:
\begin{enumerate}\itemsep=1ex
\item $((i,j),(i_1,j_1)) \in \Sem{\pi_1}{\run,1}$ and $((i,j),(i_2,j_2)) \in 
\Sem{\pi_2}{\run,1}$

\item $\smash{\run_{i_1'}^{0}(\idreg) = \run_{i_1}^{j_1}(r_1)}$ and
$\smash{\run_{i_2'}^{0}(\idreg) = \run_{i_2}^{j_2}(r_2)}$

\item $\smash{\run_{i_1'}^{0}(\idreg) \le \run_{i_2'}^{0}(\idreg)}$
\end{enumerate}}
\end{tabular}\vspace{1ex}

\begin{tabular}{rl}
$\Longleftrightarrow$ & \parbox[t]{35em}{$\forall \run \in \Runs{T}: \run,1,(i,j) \models \existsp{r_1}{r_2}{\pi_1}{\pi_2}{\le}$}\\[3ex]

$(\ast)$ & by I.H.(c), and Lemmas~\ref{lem:pathaut} and \ref{lem:piless}\\[1ex]

$(\ast\ast)$ & by I.H.(c), Lemmas~\ref{lem:pathaut} and \ref{lem:piless}, and the fact that $\pi_1$ and $\pi_2$ are unambiguous,\\&the coordinates are uniquely determined by $T$, $(i,j)$, and $\phi$
\end{tabular}
\end{center}

\item The case $\phi = \existsp{r_1}{r_2}{\pi_1}{\pi_2}{<}$ is simpler than the previous one. We just have to adapt {\textsf{3.}}\ accordingly.


\item Consider the case $\phi = \bigl(\existsp{r_1}{r_2}{\pi_1}{\pi_2}{\neq}\bigr)$.
We show (b):
\begin{center}
\begin{tabular}{rl}
& $T_\run,(i,j) \models \transl{\existsp{r_1}{r_2}{\pi_1}{\pi_2}{\neq}}$\\[2ex]

$\Longleftrightarrow$ &
\parbox[t]{35em}{$T_\run,(i,j) \models \loopform{\transl{\pi_1} \cdot (\pathaut_{r_1}^{2})^{-1} \cdot (\goleft^+ + \goright^+) \cdot \pathaut_{r_2}^{2} \cdot (\transl{\pi_2})^{-1}}$}
\end{tabular}\vspace{1ex}

\begin{tabular}{rl}
$\Longleftrightarrow$ &
\parbox[t]{35em}{
there are coordinates $(i_1,j_1),(i_2,j_2),(i_1',0),(i_2',0) \in \Coord{\run}$ such that:
\begin{enumerate}\itemsep=1ex
\item $((i,j),(i_1,j_1)) \in \Sem{\transl{\pi_1}}{T_\run}$ and $((i,j),(i_2,j_2)) \in 
\Sem{\transl{\pi_2}}{T_\run}$

\item $((i_1',0),(i_1,j_1)) \in \Sem{\pathaut_{r_1}^{2}}{T_\run}$ and $((i_2',0),(i_2,j_2)) \in \Sem{\pathaut_{r_2}^{2}}{T_\run}$

\item $i_1' \neq i_2'$
\end{enumerate}}
\end{tabular}\vspace{1ex}

\begin{tabular}{rl}
$\Longleftrightarrow$ &
\parbox[t]{35em}{
(by I.H.(c) and Lemma~\ref{lem:pathaut})\\
there are coordinates $(i_1,j_1),(i_2,j_2),(i_1',0),(i_2',0) \in \Coord{\run}$ such that:
\begin{enumerate}\itemsep=1ex
\item $((i,j),(i_1,j_1)) \in \Sem{\pi_1}{\run,1}$ and $((i,j),(i_2,j_2)) \in 
\Sem{\pi_2}{\run,1}$

\item $\smash{\run_{i_1'}^{0}(\idreg) = {\run}_{i_1}^{{j_1}}(r_1)}$ and
$\smash{\run_{i_2'}^{0}(\idreg) = {\run}_{i_2}^{{j_2}}(r_2)}$

\item $i_1' \neq i_2'$
\end{enumerate}}
\end{tabular}\vspace{1ex}

\begin{tabular}{rl}
$\Longleftrightarrow$ & \parbox[t]{35em}{$\run,1,(i,j) \models \existsp{r_1}{r_2}{\pi_1}{\pi_2}{\neq}$}
\end{tabular}
\end{center}

\item The case $\phi = \bigl(\existsp{r_1}{r_2}{\pi_1}{\pi_2}{=}\bigr)$ is almost identical. In {\textsf{3.}}, we just replace $\neq$ by $=$.

\item Consider the path formula $\pi = \test{\phi}$. Note that $\phi$ is good. We show (c):
\[\begin{array}{cl}
& \Sem{\transl{\test{\phi}}}{T_\run} = \Sem{\test{\transl{\phi}}}{T_\run}\\[1ex]
= & \{(x,x) \mid x \in \Coord{\run}: T_\run,x \models \transl{\phi}\}\\[1ex]
\smash{\stackrel{\text{I.H.(b)}}{=}} & \{(x,x) \mid x \in \Coord{\run}: \run,1,x \models \phi\}\\[1ex]
= & \Sem{\test{\phi}}{\run,1}
\end{array}\]

\item Consider $\pi = \goright$. Suppose the coordinate set of $\run$ is $\set{n} \times \setz{k}$. We show (c):
\[\begin{array}{cl}
& \Sem{\transl{\goright}}{T_\run} = \Sem{\goright + \test{\neg\existspath{\goright}}\goleft^\ast\test{\neg\existspath{\goleft}}}{T_\run}\\[1ex]
= & \{((i,j),(i+1,j)) \mid (i,j) \in \set{n-1} \times \setz{k}\} \cup
\{((n,j),(1,j)) \mid j \in \setz{k}\}\\[1ex]
= & \Sem{\goright}{\run,1}\\[1ex]
\end{array}\]

\item The regular operations as well as $\goup$ and $\godown$ are obvious, and the case $\goleft$ is symmetric to $\goright$.

\end{itemize}


\section{Proof of Lemma~\ref{lem:satisf}}

Let us prove (a).\medskip

\noindent
($\Rightarrow$): Suppose $\DA \models \Phi = \Allrings{\locform}$. Let $T \in L(\lcpdl_\DA)$. By Lemma~\ref{lem:datapdl}, there is a run $\run$ of $\DA$ such that $T_\run=T$. Moreover, since $\DA \models \Phi$, all runs $\run$ of $\DA$ satisfy $\run,1,(1,0) \models \phi$. This applies, in particular, to all runs $\run$ such that $T_\run = T$. By Lemma~\ref{lem:lcpdl}, we have $T,(1,0) \models \transl{\phi}$. We conclude $L(\psi_\DA \wedge \neg\transl{\phi}) = \emptyset$.

\medskip

\noindent
($\Leftarrow$): Suppose $\DA \not\models \Allrings{\locform}$. Then, there are a ring $\ring=(n:\ldots)$, an $\ring$-run $\run$ of $\DA$, and a process $m \in \set{n}$ such that $\run,m,(m,0) \not\models \phi$. Since $\phi$ cannot distinguish isomorphic rings, we can shift $\ring$ until $m$ ``arrives'' on position $1$. Thus, there are $\ring'=(n:\ldots)$ and an $\ring'$-run $\run'$ of $\DA$ such that $\run',1,(1,0) \not\models \phi$. By Lemma~\ref{lem:lcpdl}, $T_{\run'},(1,0) \not\models \transl{\phi}$ and, therefore, $T_{\run'},(1,0) \models \neg\transl{\phi}$. Due to Lemma~\ref{lem:datapdl}, we also have $T_{\run'},(1,0) \models \psi_\DA$. we conclude $L(\lcpdl_\DA \wedge \neg\transl{\locform}) \neq \emptyset$.

\bigskip

Part (b) is shown in exactly the same way, restricting the height of a table and length of a run by the given bound  $\bound$.

  
\end{document}